\documentclass[final,onefignum,onetabnum]{siamart171218}
\usepackage{xspace}
\usepackage{amsmath, amssymb,amsfonts,enumerate,hyperref,comment,url,thm-restate}
\usepackage{mathtools, tabularx}
\usepackage{tikz}
\usepackage{algorithm}
\usepackage[noend]{algpseudocode}
\usepackage{microtype}
\usepackage{subcaption}
\usepackage{enumerate}
\usepackage{booktabs}
\usepackage{graphicx,comment}

\newsiamthm{claim}{Claim}

\def\final{1}  
\def\iflong{\iffalse}
\ifnum\final=0  
\newcommand{\snote}[1]{[{\tiny Sahand: \bf #1}]\marginpar{*}}
\newcommand{\knote}[1]{[{\tiny Karthik: \bf #1}]\marginpar{*}}
\newcommand{\todo}[1]{[{\tiny TODO: \bf #1}]\marginpar{*}}
\else 
\newcommand{\snote}[1]{}
\newcommand{\knote}[1]{}
\newcommand{\todo}[1]{}
\fi  

\newcommand{\NP}{\ensuremath{\mathsf{NP}}}

\usetikzlibrary{positioning}
\usetikzlibrary{arrows.meta}
\usetikzlibrary{arrows}
\usetikzlibrary{fit}
\tikzstyle{vertex}=[circle, draw, fill=black, inner sep=0pt, minimum size=6pt]
\tikzstyle{thickedge}=[line width=2pt]
\tikzstyle{thinedge}=[thin]

\newlength\nodedistance
\newlength\longnodedistance
\newlength\horiznodedistance
\newlength\vertnodedistance

\newcommand{\card}[1]{\left|#1\right|}
\newcommand{\set}[1]{\left\{#1\right\}}

\newcommand{\suchthat}{:}
\newcommand{\Z}{\ensuremath{\mathbb{Z}}\xspace}
\newcommand{\R}{\ensuremath{\mathbb{R}}\xspace}
\newcommand{\Zbb}{\ensuremath{\mathbb{Z}}\xspace}
\newcommand{\Rbb}{\ensuremath{\mathbb{R}}\xspace}
\newcommand{\calI}{\mathcal{I}}
\newcommand{\calR}{\mathcal{R}}

\newcommand{\prbname}{\textsc{MinOddBlocker}\xspace}
\newcommand{\vc}{\textsc{Vertex Cover}\xspace}

\newcommand{\oddpathedgeblocker}{\textsc{$\{s,t\}$-Odd\-Path\-Edge\-Blocker}}
\newcommand{\oddpathnodeblocker}{\textsc{$\{s,t\}$-Odd\-Path\-Node\-Blocker}}

\newcommand{\diroddpathedgeblocker}{\textsc{$(s\rightarrow t)$-Odd\-Path\-Edge\-Blocker}}

\newcommand{\diroddpathnodeblocker}{\textsc{$(s\rightarrow t)$-Odd\-Path\-Node\-Blocker}}
\newcommand{\mwc}{\textsc{MultiwayCut}\xspace}
\newcommand{\undirpathblocker}{\textsc{$\{s,t\}$-Odd\-Path\-Edge\-Blocker}\xspace}
\newcommand{\undirpathnodeblocker}{\textsc{$\{s,t\}$-Odd\-Path\-Node\-Blocker}\xspace}
\newcommand{\omwnodec}{\textsc{OddMultiwayNodeCut}\xspace}
\newcommand{\omwedgec}{\textsc{OddMultiwayEdgeCut}\xspace}
\newcommand{\pmwnodec}{\textsc{ParityMultiwayNodeCut}\xspace}
\newcommand{\pmwedgec}{\textsc{ParityMultiwayEdgeCut}\xspace}
\newcommand{\dagoddmultiwaycut}{\textsc{OddMultiwayNodeCut}\xspace}

\newcommand{\etal}{\textit{et al}.}

\DeclareMathOperator{\torso}{\texttt{ParityTorso}}

\newcommand{\MinBipartization}{{\sc OddCycleTransversal}}

\DeclareMathOperator{\poly}{poly}
\DeclareMathOperator{\ShadowContainer}{\texttt{ShadowContainer}}
\DeclareMathOperator{\ReverseShadowContainer}{\texttt{ReverseShadowContainer}}

\newcommand{\defparproblem}[4]{
    \bigskip
    \noindent
    \fbox{
        \begin{minipage}{.96\linewidth}
            {\sc #1} \hfill Parameter: #2\\[2pt]
            \smallskip
            \noindent
            \begin{tabular}{@{}l@{ }l}
                \emph{Input:} & \begin{minipage}[t]{\linewidth-\widthof{Input:\ \ }}
                    #3
                \end{minipage}\\[2pt]
                \emph{Task:} & \begin{minipage}[t]{\linewidth-\widthof{Input:\ \ }}
                    #4
                \end{minipage}
            \end{tabular}
        \end{minipage}
    }
    \medskip
}

\title{Odd Multiway Cut in Directed Acyclic Graphs\thanks{A preliminary version of these results appeared in the proceedings of the 12th International Symposium on Parameterized and Exact Computation (IPEC 2017).}}

\author{
Karthekeyan Chandrasekaran\thanks{University of Illinois, Urbana-Champaign.  Email: \email{karthe@illinois.edu}.}
\and
Matthias Mnich\thanks{Maastricht University, Maastricht, The Netherlands \emph{and} Universit{\"a}t Bonn, Bonn, Germany. Email: \email{mmnich@uni-bonn.de}. Supported by ERC Starting Grant 306465 (BeyondWorstCase) and DFG grant MN 59/1-1.} 
\and
Sahand Mozaffari\thanks{University of Illinois, Urbana-Champaign.  Email: \email{sahandm2@illinois.edu}.}
}

\headers{Odd Multiway Cut in Directed Acyclic Graphs}{K.~Chandrasekaran, M.~Mnich, S.~Mozaffari}

\begin{document}
\maketitle              

\begin{abstract}
We investigate the odd multiway node (edge) cut problem where the input is a graph with a specified collection of terminal nodes and the goal is to find a smallest subset of non-terminal nodes (edges) to delete so that the terminal nodes do not have an odd length path between them. In an earlier work, Lokshtanov and Ramanujan showed that both odd multiway node cut and odd multiway edge cut are fixed-parameter tractable (FPT) when parameterized by the size of the solution in \emph{undirected graphs}. In this work, we focus on directed acyclic graphs (DAGs) and design a fixed-parameter algorithm. 
Our main contribution is 
a broadening 
of the shadow-removal framework to address \emph{parity problems in DAGs}. 
We complement our FPT results with tight approximability as well as polyhedral results for $2$ terminals in DAGs. Additionally, we show inapproximability results for odd multiway \emph{edge} cut in undirected graphs even for $2$ terminals. 
\end{abstract}



\section{Introduction}
\label{sec:intro}
In the classic $\{s,t\}$-cut problem, the goal is to delete the smallest number of edges so that the resulting graph has no path between $s$ and $t$.
A natural generalization of this problem is the multiway cut problem, where the input is a graph with a specified set of terminal nodes and the goal is to delete the smallest number of non-terminal nodes/edges so that the terminals cannot reach each other in the resulting graph.  
In this work, we consider a parity variant of the multiway cut problem.
A path\footnote{We emphasize that the term \emph{paths} refers to simple paths and not walks.
    This distinction is particularly important in parity-constrained settings, because the existence of a walk with an odd number of edges between two nodes $s$ and $t$ does not imply the existence of an odd-path between $s$ and $t$. This is in contrast to the non-parity-constrained settings where the existence of a walk between $s$ and $t$ implies the existence of a path between $s$ and $t$.
} is an \emph{odd-path} (\emph{even-path}) if the number of edges in the path is odd (even).
In the 
\omwnodec (similarly, \omwedgec), the input is a graph with a collection of terminal nodes and 
the goal is to delete the smallest number of non-terminal nodes (edges) so that the resulting graph has no odd-path between the terminals.
This is a generalization of \oddpathnodeblocker\ (and similarly, \oddpathedgeblocker), which is the problem of finding a minimum number of nodes (edges) that are disjoint from $s$ and $t$ that cover all $s-t$ odd-paths. 

Covering and packing paths has been a topic of intensive investigation in graph theory as well as polyhedral theory. 
Menger's theorem gives a perfect duality relation for covering $s-t$ paths: 
the minimum number of nodes (edges) that cover all $s-t$ paths is equal to the maximum number of node-disjoint (edge-disjoint) $s-t$ paths.
However, packing paths of restricted kinds has been observed to be a difficult problem in the literature.
One special case is when the paths are required to be of odd-length for which many structural results exist \cite{ChudnovskyEtAl2006,Schrijver2003,ibrahimpur2017min}. 
In this work, we study the problem of covering $s-t$ odd-paths and more generally all odd-paths between a given collection of terminals. 

Covering $s-t$ odd-paths in undirected graphs has been explored in the literature from the perspective of polyhedral theory---e.g., see Chapter 29 in Schrijver's book \cite{Schrijver2003}.
Given an undirected graph $G=(V,E)$ with distinct nodes $s,t\in V$ and non-negative edge lengths, we may find a shortest length $s-t$ odd-path in polynomial time. Edmonds gave a polynomial-time algorithm for the shortest length $s-t$ odd-path problem by reducing it to the minimum-weight perfect matching problem \cite{Edmonds1965,Edmonds1965b,LaPaughPapadimitriou1984}.
However, as observed by Schrijver and Seymour \cite{SchrijverSeymour1994}, his approach of reducing to a matching problem does not extend to address other fundamental problems about $s-t$ odd-paths. One such fundamental problem is the  \oddpathedgeblocker\ problem. 
Towards investigating \oddpathedgeblocker, Schrijver and Seymour \cite{SchrijverSeymour1994} considered the following polyhedron: 
\[
  \mathcal{P}^{\text{odd-cover}}:=\left\{x\in \R^{E}_+: \sum_{e\in P}x_e\ge 1\ \forall\ \text{$s-t$ odd-path $P$ in $G$}\right\} \enspace .
\]
This leads to a natural integer programming formulation of \oddpathedgeblocker: 
$\min$ $\left\{\sum_{e\in E}x_e: x\in \mathcal{P}^{\text{odd-cover}}\cap \Z^E\right\}$. 
By Edmonds' algorithm, we have an efficient separation oracle for $\mathcal{P}^{\text{odd-cover}}$ and hence there exists an efficient algorithm to optimize over $\mathcal{P}^{\text{odd-cover}}$ using the ellipsoid algorithm \cite{GrotschelEtAl1993}.
It was known that the extreme points of $\mathcal{P}^{\text{odd-cover}}$ are not integral.
Cook and Seb\H{o} conjectured that all extreme points of $\mathcal{P}^{\text{odd-cover}}$ are half-integral which was later shown by Schrijver and Seymour \cite{SchrijverSeymour1994}. 
Schrijver and Seymour's work also gave a min-max relation for the maximum fractional packing of $s-t$ odd-paths.
However, their work does not provide algorithms to address \oddpathedgeblocker. In fact, even the computational complexity of \oddpathedgeblocker\ has been open.

In this work, we undertake a comprehensive study of \omwnodec\ and \omwedgec\ in directed acyclic graphs (DAGs). In addition to approximability, we focus on fixed-parameter tractability. Fixed-parameter algorithms have served as an alternative approach to address \NP-hard problems~\cite{CyganEtAl2015}.
A fixed-parameter algorithm for a problem decides all the problem's instances of size $n$ in time $f(k)\cdot n^{O(1)}$ for some computable function $f$, where $k$ is some integer parameter.
Fixed-parameter algorithms for cut problems have provided novel insights into the connectivity structure of graphs \cite{CyganEtAl2015}.
The notion of \emph{important separators} and the \emph{shadow-removal technique} have served as the main ingredients in the design of fixed-parameter algorithms for numerous cut problems \cite{ChitnisEtAl2013,Marx2006,MarxRazgon2014,BringmannEtAl2016,MarxEtAl2013}.
Our work also builds upon the shadow-removal technique to design fixed-parameter algorithms but differs from known applications substantially owing to the parity constraint.
Parity-constrained cut problems have attracted much interest in the parameterized complexity community~\cite{LokshtanovRamanujan2012,LokshtanovEtAl2017,ReedEtAl2004,KratschWahlstrom2014} mainly due to their challenging nature: indeed, designing fixed-parameter algorithms for parity-constrained cut problems sparked the development of new and powerful techniques~\cite{ReedEtAl2004,KratschWahlstrom2014}.

\subsection{Our contributions}

The main focus of this work is \omwnodec\ in directed acyclic graphs (DAGs).
Before describing the reason for focusing on the subfamily of DAGs among directed graphs, we note that \omwnodec\ and \omwedgec\ are equivalent in directed graphs by standard reductions (e.g., see Lemma~\ref{lem:omwedgec-to-nodec}). 
The reason we focus on the subfamily of DAGs as opposed to all directed graphs is 
due to the following fact: it is \NP-complete to verify if a given directed graph has an $s\rightarrow t$ odd-path (e.g., see LaPaugh-Papadimitriou~\cite{LaPaughPapadimitriou1984}). This fact already illustrates a stark contrast in the complexity between \oddpathedgeblocker\ in undirected graphs and \diroddpathedgeblocker\ in directed graphs: while verifying feasibility of a solution to 
\oddpathedgeblocker\ 
in undirected graphs can be done in polynomial-time, verifying feasibility of a solution to 
\diroddpathedgeblocker\ 
in directed graphs is \NP-complete.
However, there exists a polynomial time algorithm to verify if a given directed \emph{acyclic} graph (DAG) has an $s\rightarrow t$ odd-path (e.g., see Lemma~\ref{lem:odd-path-decision}).
For this reason, we restrict our focus to DAGs. 

Our main contribution is a fixed-parameter algorithm for \omwnodec\ in DAGs.
We complement the fixed-parameter algorithm by showing \NP-hardness and tight approximability results as well as polyhedral results for the two terminal variant, namely \diroddpathnodeblocker, in DAGs.

In addition to the above results for DAGs, we also show \NP-hardness and an inapproximability result for \oddpathedgeblocker\ in undirected graphs. 

\subsection{Related work}
\label{sec:related-work}
We are not aware of any work on this problem in directed graphs.
We describe the known results in undirected graphs.
A simple reduction\footnote{Given an instance $G$ of {\sc Vertex Cover}, introduce two new nodes $s$ and $t$ that are adjacent to all nodes in~$G$ to obtain a graph $H$.
Then a set $S\subseteq V(G)$ is a vertex cover in $G$ if and only if $S$ is a feasible solution to \oddpathnodeblocker\ in $H$.} from {\sc Vertex Cover} 
shows that \oddpathnodeblocker\ in undirected graphs is \NP-hard and does not admit a $(2-\varepsilon)$-approximation for $\varepsilon>0$ assuming the Unique Games Conjecture~\cite{KhotRegev2008}.
These hardness results also hold for \omwnodec.
The most relevant results to this work are that of Lokshtanov and Ramanujan~\cite{LokshtanovRamanujan2012,Ramanujan2013}. 
They studied an extension of \omwnodec\ and \omwedgec that they termed as \pmwnodec (and \pmwedgec)---the input is an undirected graph and two subsets of terminals $T_e$ and $T_o$ and the goal is to find the smallest number of non-terminal nodes (edges) so that every node $u\in T_o$ has no odd-path to any node in $T_e\cup T_o$ and every node $u\in T_e$ has no even-path to any node in $T_e\cup T_o$.
Lokshtanov and Ramanujan designed a fixed-parameter algorithm for \pmwnodec by reducing the problem to \omwnodec\ and designing a fixed-parameter algorithm for \omwnodec.
However, their algorithmic techniques work only for undirected graphs and do not extend for \omwnodec\ in DAGs. 

Lokshtanov and Ramanujan also showed that \omwedgec is \NP-hard in undirected graphs for three terminals.
However, their reduction is not an approximation-preserving reduction.
Hence the approximability of \omwedgec\ in undirected graphs merits careful investigation.
In particular, the complexity of \omwedgec\ in undirected graphs even for the case of two terminals is open in spite of existing polyhedral work in the literature~\cite{SchrijverSeymour1994} for this problem. 

The subset odd cycle transversal problem (\textsc{SubsetOCT}) generalizes the \omwnodec problem in undirected graphs.
Here, the input is an undirected graph $G$ with a subset of vertices $T$ and the goal is to determine a smallest subset of vertices that intersects every odd cycle containing a vertex from $T$. Fixed-parameter algorithms for \textsc{SubsetOCT} are also known in the literature~\cite{LokshtanovEtAl2017b}. 


\subsection{Results}
\noindent{\bf Directed acyclic graphs.}
We recall that \omwnodec\ and \omwedgec\ are equivalent in DAGs by standard reductions.
Hence, all of the following results for DAGs hold for both problems.

The following is our main result. 
\begin{theorem}
\label{thm:omwnodec-dag}
  \omwnodec\ and \omwedgec in DAGs can be solved in $2^{O(k^2)}\cdot n^{O(1)}$ time, where $k$ is the size of the optimal solution and $n$ is the number of nodes in the input graph. 
\end{theorem}

We briefly remark on the known techniques to illustrate the challenges in designing the fixed-parameter algorithm for \omwnodec\ in DAGs.
To highlight the challenges, we will focus on the case of $2$ terminals, namely \diroddpathnodeblocker\ in DAGs.

\medskip
\noindent
{\bf Remark 1.} It is tempting to design a fixed-parameter algorithm for \diroddpathnodeblocker\ by suitably modifying the definition of \emph{important separators} to account for parity and then attempt to use the shadow-removal framework for directed graphs \cite{ChitnisEtAl2013}. 
However, it is unclear how the suitable modifications can exploit the acyclic property of the input directed graph. 

\medskip
\noindent
{\bf Remark 2.} The next natural attempt is to rely on the fixed-parameter algorithm for multicut in DAGs by Kratsch et al.~\cite{KratschEtAl2015}.
Their technique crucially relies on reducing the degrees of the source terminals by suitably branching to create a small number of instances. 
On the one hand, applying their branching rule directly to reduce the degree of $s$ in \diroddpathnodeblocker\ will blow up the number of instances in the branching.
On the other hand, it is unclear how to modify their branching rule to account for parity.


\medskip
\noindent
Given the difficulties mentioned in the above two remarks, our algorithm builds upon the shadow-removal framework.
We exploit the acyclic property of the input directed graph to reduce the instance to an instance of {\sc Odd Cycle Transversal}. In {\sc Odd Cycle Transveral}, the goal is to remove the smallest number of nodes to make an undirected graph bipartite.
{\sc Odd Cycle Transversal} is fixed-parameter tractable when parameterized by the number of removed nodes.
We view our technique as an illustration of the broad-applicability of the shadow-removal framework.

\medskip

We complement our fixed-parameter algorithm in Theorem \ref{thm:omwnodec-dag} with tight approximability results for the special case of $2$ terminals.
We refer the reader to Table \ref{table:approx} for a summary of the complexity and approximability results.
Unlike the case of undirected graphs where there is still a gap in the approximability of both \oddpathedgeblocker\ and \oddpathnodeblocker, we present tight approximability results for both \diroddpathedgeblocker\ and \diroddpathnodeblocker in DAGs:
\begin{theorem}
\label{thm:diroddpathnodeblocker-approx}
  We have the following inapproximability and approximability results:
  \begin{enumerate}
    \item[(i)] \diroddpathnodeblocker\ in DAGs is \NP-hard, 
    and has no efficient $(2-\varepsilon)$-approximation for any $\varepsilon>0$ assuming the Unique Games Conjecture.
    \item[(ii)] There exists an efficient $2$-approximation algorithm for \diroddpathnodeblocker\ in DAGs.
  \end{enumerate}
\end{theorem}

We emphasize that our $2$-approximation algorithm for \diroddpathedgeblocker\ mentioned in Theorem~\ref{thm:diroddpathnodeblocker-approx} is a combinatorial algorithm and not LP-based.
We note that Schrijver and Seymour's result~\cite{SchrijverSeymour1994} that all extreme points of $\mathcal{P}^{\text{odd-cover}}$ are half-integral holds only in undirected graphs and fails in DAGs---see Theorem~\ref{thm:odd-path-blocker-polyhedron-not-half-integral} below.
Consequently, we are unable to design a $2$-approximation algorithm using the extreme point structure of the natural LP-relaxation of the path-blocking integer program.
Instead, our approximation algorithm is combinatorial in nature.
The correctness argument of our algorithm also shows that the integrality gap of the LP-relaxation of the path-blocking integer program is at most $2$ in DAGs.

\begin{theorem}
  \label{thm:odd-path-blocker-polyhedron-not-half-integral}
  The odd path cover polyhedron given by 
  \[
    \mathcal{P}^{\text{odd-cover-dir}}:=\left\{x\in \R_+^E: \sum_{e\in P}x_e\ge 1\ \forall\ \text{$s\rightarrow t$ odd-path $P$ in $D$}\right\}
  \]
  for directed acyclic graphs $D=(V,E)$ is not necessarily half-integral.
\end{theorem}

\begin{table}[ht]
  \begin{center}
    \small{
    \begin{tabular}{ccc}
      \toprule
      {\bf Problem} & {\bf Undirected graphs} & {\bf DAGs}\\
      \midrule
      \oddpathnodeblocker & {\color{gray} $(2-\varepsilon)$-inapprox} & $2$-approx (Thm. \ref{thm:diroddpathnodeblocker-approx})\\
	                      &                                           & $(2-\varepsilon)$-inapprox (Thm \ref{thm:diroddpathnodeblocker-approx})\\
      \midrule
      \oddpathedgeblocker & {\color{gray} LP is half-integral \cite{SchrijverSeymour1994}} & LP is NOT half-integral \\
                          &                                                & (Thm. \ref{thm:odd-path-blocker-polyhedron-not-half-integral})\\
                          & {\color{gray} $2$-approx \cite{SchrijverSeymour1994}}          & $2$-approx (Thm. \ref{thm:diroddpathnodeblocker-approx})\\
                          & $(\frac{6}{5}-\varepsilon)$-inapprox (Thm. \ref{thm:odd-path-edge-blocker-hardness}) & $(2-\varepsilon)$-inapprox (Thm \ref{thm:diroddpathnodeblocker-approx})\\
     \midrule
     \omwedgec & {\color{gray} \NP-hard for $3$ terminals \cite{LokshtanovRamanujan2012}} & $2^{O(k^2)}\cdot \textnormal{poly}(n)$ (Thm. \ref{thm:omwnodec-dag}) \\
       & $(\frac{6}{5}-\varepsilon)$-inapprox for $2$ terminals & \\
       & (Thm \ref{thm:odd-path-edge-blocker-hardness}) & \\
      \midrule
      \omwnodec & & $2^{O(k^2)}\cdot \textnormal{poly}(n)$ (Thm. \ref{thm:omwnodec-dag}) \\
      \bottomrule
    \end{tabular}
    \caption{Complexity and approximability. Text in gray refers to known results while text in black refers to the results from this work.\label{table:approx}}
    }
  \end{center}
\end{table}

\noindent{\bf Undirected graphs.} We next turn our attention to undirected graphs.
As mentioned in Section~\ref{sec:related-work}, the problem \oddpathnodeblocker\ is \NP-hard and does not admit a $(2-\varepsilon)$-approximation assuming the Unique Games Conjecture. 
We are unaware of a constant factor approximation for \oddpathnodeblocker. 
For \oddpathedgeblocker, the results of Schrijver and Seymour~\cite{SchrijverSeymour1994} show that the LP-relaxation of a natural integer programming formulation of \oddpathedgeblocker\ is half-integral and thus leads to an efficient $2$-approximation algorithm. 
However, the complexity of \oddpathedgeblocker\ was open.
We address this gap in complexity by showing the following \NP-hardness and inapproximability results. 

\begin{restatable}{theorem}{undirOddPathBlockerHardness}
\label{thm:odd-path-edge-blocker-hardness}
  \oddpathedgeblocker\ is \NP-hard and has no efficient $(6/5-\varepsilon)$-approximation assuming the Unique Games Conjecture.  
\end{restatable}

\noindent {\bf Organization.} We summarize the preliminaries in Section \ref{sec:prelims}.
We devise the fixed-parameter algorithms for DAGs (Theorem~\ref{thm:omwnodec-dag}) in Section~\ref{sec:DAG-FPT}.
We complement with approximability results for DAGs (Theorems~\ref{thm:diroddpathnodeblocker-approx} and \ref{thm:odd-path-blocker-polyhedron-not-half-integral}) in Section~\ref{sec:DAG-approx-hardness}.
Next, we focus on undirected graphs and present the inapproximability result (Theorem~\ref{thm:odd-path-edge-blocker-hardness}) in Section \ref{sec:undir-approximability}.
We conclude by discussing a few open problems in Section~\ref{sec:discussion}.

\subsection{Preliminaries}
\label{sec:prelims}
Let $G$ be a (directed) graph with vertex set $V(G)$ and edge set $E(G)$.
For single vertices $v\in V(G)$, we will frequently use $v$ instead of $\set{v}$.
For a subset $W\subseteq V(G)$, \emph{a $W$-path in $G$} is a path with both of its end nodes in $W$.

For a directed acyclic graph $G$ and node sets $T$, $V^{\infty} \subseteq V(G)$ where $T\subseteq V^{\infty}$, 
an \emph{odd multiway cut} in $G$ is a set $M \subseteq V(G) \setminus V^\infty$ of nodes that intersects every odd $T$-path in $G$.
We refer to elements of $T$ as \emph{terminals}, elements of $V(G)\setminus T$ as \emph{non-terminals}, and elements of $V^{\infty}$ as \emph{protected nodes}.

We restate the problem of \omwnodec\ in DAGs to set the notation.


\defparproblem{\omwnodec in DAGs}{$k$}{A DAG $G$ with a set $V^\infty\subseteq V(G)$ of protected nodes and set $T\subseteq V^\infty$ of terminal nodes, and an integer $k\in\Z_+$.}{Verify if there is an odd multiway node cut in $G$ of size at most $k$.} 


For subsets $X$ and $Y$ of $V(G)$ we say that $M \subseteq V(G) \setminus V^\infty$ is an \emph{$X \rightarrow Y$ separator} in $G$ when $G \setminus M$ has no path from $X$ to $Y$.
The set of nodes that can be reached from a node set $X$ in $G$ is denoted by $\calR_G(X)$.
We note that $\calR_G(X)$ always includes $X$.

We define the \emph{forward shadow} of a node set $M$ to be $f_G(M):=V(G \setminus M) \setminus \calR_{G \setminus M}(T)$, i.e., the set of nodes $v$ such that there is no $T \rightarrow v$ path in $G$ disjoint from $M$.
Similarly, the \emph{reverse shadow} of $M$, denoted $r_G(M)$, is the set of nodes $v$ from which there is no path to $T$ in~$G \setminus M$.
Equivalently, the reverse shadow is $f_{G^\text{rev}}(M)$, where $G^\text{rev}$ is the graph obtained from~$G$ by reversing all the edge orientations.
We refer to the union of the forward and the reverse shadow of $M$ in $G$, as \emph{shadow} of $M$ in $G$ and denote it by $s_G(M)$.
A set $M \subseteq V(G)$ is \emph{thin}, if every node $v \in M$ is not in $r_G(M \setminus \set{v})$.

We need the notion of important separators \cite{Marx2006}.
An $X \rightarrow Y$ separator $M'$ is said to \emph{dominate} another $X \rightarrow Y$ separator $M$, if $\card{M'} \leq \card{M}$ and $\calR_{G \setminus M}(X) \subsetneq \calR_{G \setminus M'}(X)$. A minimal $X \rightarrow Y$ separator that is not dominated by any other separator is called an \emph{important $X \rightarrow Y$ separator}.

For a directed graph $G$, its underlying undirected graph $\langle G\rangle$ is the undirected graph obtained from $G$ by dropping the edge orientations.
In an undirected graph $H$ with protected nodes~$V^\infty$, an \emph{odd cycle transversal} is a set $U \subseteq V(H) \setminus V^\infty$ of nodes such that $H \setminus U$ is bipartite.
The problem of finding a minimum odd cycle transversal in a given instance $(H,V^\infty)$ is the \MinBipartization{} problem.
This problem is \NP-hard, but admits fixed-parameter algorithms when parameterized by the size $k$ of an optimal solution.
The asymptotically fastest fixed-parameter algorithm for \MinBipartization{} in terms of $k$ is due to Lokshtanov et al.~\cite{LokshtanovEtAl2014}; it runs in time $2.32^k\cdot n^{O(1)}$, and is based on linear programming techniques.
While their algorithm does not allow for protected nodes, the problem \MinBipartization\ with protected nodes can be reduced to \MinBipartization\ without protected nodes by iteratively replacing each protected node with $k+1$ nodes and connecting them to the same set of neighbors as the original node.
We thus have:
\begin{proposition}
  \label{thm:fastoct}
  There is an algorithm that, given an undirected $n$-node graph $H$, a set $V^\infty\subseteq V(H)$ of protected nodes and an integer $k$, decides if $H$ admits an odd cycle transversal of size at most $k$ that is disjoint from $V^\infty$, and if so, returns one. Moreover, the algorithm runs in time $2.32^k\cdot n^{O(1)}$.
\end{proposition}
We will use \MinBipartization$(H, V^\infty, k)$ to denote the procedure that implements this fixed-parameter algorithm for the input graph $H$ with protected nodes $V^{\infty}$ and parameter~$k$. 
\section{Fixed-parameter tractability of \omwnodec\ in DAGs}
\label{sec:DAG-FPT}
To solve \dagoddmultiwaycut in DAGs, we will use the shadow-removal technique introduced by Chitnis et al.~\cite{ChitnisEtAl2013}.
We will reduce the problem to the \MinBipartization{} problem in undirected graphs, which is a fixed-parameter tractable problem when parameterized by the solution size.
We begin by arguing about ``easy'' instances, where we define an instance $(G,V^{\infty},T,k)$ as \emph{easy} if it has a solution $M$ (of size at most $k$) where every node $v\in s_G(M)$ has total degree at most one in $G \setminus M$, provided that it has some solution (of size at most $k$) at all.

\subsection{Easy instances}
\begin{theorem}
\label{thm:dag-easy-instances}
%
  There is an algorithm that, given any easy instance $(G,V^{\infty},T,k)$ of \omwnodec where $G$ is a DAG, finds a solution of size at most $k$ in time $2.32^k\cdot n^{O(1)}$, where $n$ is the number of nodes in the input graph $G$. 
\end{theorem}
\begin{proof}
  Let $(G,V^{\infty},T,k)$ be an instance of \dagoddmultiwaycut. 
  Let $\langle G\rangle$ denote the undirected graph obtained from $G$ by dropping the orientations of the edges in $G$. 
  We show the following equivalence: a set $M\subseteq V\setminus V^{\infty}$ with the property as in the statement is a solution if and only if 
  $\langle G\rangle \setminus M$ is bipartite with a bipartition $(A,B)$ such that $T\subseteq A$. 

  Suppose $\langle G\rangle \setminus M$ is bipartite with a bipartition $(A,B)$ such that $T\subseteq A$.
  In a bipartite graph, every two end-nodes of any odd path are necessarily in different parts.
  Hence, there is no odd $T$-path in $\langle G\rangle \setminus M$.
  Thus, there is no odd $T$-path in $G\setminus M$.
  Hence, the set $M$ is a solution for the \omwnodec\ instance $(G,V^{\infty},T,k)$. 

  Suppose the solution $M$ has the property mentioned in the statement of the theorem. 
  Let $U := V(G \setminus M) \setminus s_G(M)$. 
  Define
  \begin{align*}
    A &:= \set{x \in U \suchthat \text{there is an even $T \rightarrow x$ path in $G \setminus M$}} \text{ and}\\
    B &:= \set{x \in U  \suchthat \text{there is an odd $T \rightarrow x$ path in $G \setminus M$}}.
  \end{align*}
  It follows from the definition of the shadow that every node in $U$ has a path $P_1$ from $T$ in $G \setminus M$.
  Therefore, every node of $U$ is in $A \cup B$.
  Also by definition, every node $v$ in $U$ has a path $P_2$ to $T$ in $G \setminus M$.
  The parity of every $T \rightarrow v$ path has to be the same as the parity of $P_2$, because the concatenation of a $T\rightarrow v$ path and a $v\rightarrow T$ path in $G \setminus M$ is a $T$-path in $G \setminus M$ and therefore must be even.
  We note that such a concatenation cannot be a cycle since $G$ is acyclic.
  Thus, no node of $U$ is in both $A$ and~$B$.
  Hence, we have that $(A,B)$ is a partition of $U$.

  We observe that there cannot be an edge from a node $v$ in $A$ to a node $u$ in $A$, as otherwise the concatenation of the even $T \rightarrow v$ path $Q_1$ with the edge $v \rightarrow u$ is an odd $T \rightarrow u$ path in~$G \setminus M$ which means $u \in B$.
  This contradicts our conclusion about $A$ and $B$ being disjoint.
  By a similar argument, there is no edge between any pair of nodes in $B$.
  Thus, the subgraph of~$G$ induced by $A$ and $B$ are independent sets respectively.
  Hence $\langle G\rangle[A \cup B]$ is a bipartite graph.
  Furthermore, $(A, B)$ is a bipartition of $\langle G\rangle[A \cup B]$ with every node of $T$ in $A$.
  By assumption, the degree of every node $x \in s_G(M)$ is at most one.
  Therefore, $x$ has neighbors in at most one of $A$ and $B$.
  Thus, we can extend the bipartition $(A, B)$ of $\langle G\rangle[A \cup B]$ to a bipartition $(A', B')$ of $\langle G\rangle \setminus M$ 
  as follows: denote $H:=\langle G\rangle[A\cup B]$; repeatedly pick a node $x\in s_G(M)\setminus V(H)$ with a neighbor in $H$, include $x$ in a part ($A$ or $B$) in which $x$ has no neighbor and update $A$, $B$ and $H$. 

  Hence, if the given instance has a solution $M$ of size at most $k$ such that every node $v\in s_G(M)$ has total degree at most one, then such a solution can be found by the fixed-parameter algorithm for \MinBipartization.
  To ensure that the terminal nodes will be in the same part, we introduce a new protected node into the graph and connect it to every terminal node.
  This approach is described in Algorithm~\ref{alg:solve easy instance}.
  \begin{algorithm}
    \caption{$\mathtt{SolveEasyInstance}$\label{alg:solve easy instance}}
    \begin{algorithmic}[1]
      \State \textbf{Input:} A DAG $G$ with a set $V^\infty\subseteq V(G)$ of protected nodes and a set $T\subseteq V^\infty$ of terminals, and an integer $k\in \Z_+$.
      \State \textbf{Output:} A minimum odd multiway cut for $(G,V^{\infty},T,k)$. 
      \State $G_1 \gets$ the underlying undirected graph of $G$, i.e., $\langle G\rangle$.
      \State Let $G_2$ be the graph obtained from $G_1$ by introducing a new node $x$ and connecting it to every node in $T$.
      \State $N \gets$ \MinBipartization$(G_2, V^\infty \cup \set{x}, k)$ \label{line:solve undirected}
      \State \Return $N$
    \end{algorithmic}
  \end{algorithm}
  All steps in Algorithm~\ref{alg:solve easy instance} can be implemented to run in polynomial time except Step~\ref{line:solve undirected}.
  By Proposition~\ref{thm:fastoct}, Step~\ref{line:solve undirected} can run in time $2.32^k\cdot n^{O(1)}$.
\end{proof}

We will use the name $\mathtt{SolveEasyInstance}$ to refer to the algorithm of Theorem~\ref{thm:dag-easy-instances}.
Theorem \ref{thm:dag-easy-instances}
suggests that the existence of a solution $M$ of size at most $k$, such that every node $v \in s_G(M)$ has total degree at most one, is a useful property in an instance of \dagoddmultiwaycut.
However, it is not necessarily the case that some solution of size at most $k$ always has this property.
Our aim now is to reduce the given arbitrary instance $(G, V^\infty, T, k)$ to another instance that has such a solution or determine that no solution of size at most $k$ exists.
For this purpose, we define the operation parity-preserving torso on DAGs, as follows.

\subsection{Parity-preserving torso}
\label{sec:dag-parity-torso}
The parity-preserving torso operation was introduced by Lokshtanov and Ramanujan~\cite{LokshtanovRamanujan2012} for undirected graphs.
We extend it in a natural fashion for DAGs. 
\begin{definition}[Parity-preserving torso.]
  Let $G$ be a DAG and $Z$ be a subset of $V(G)$.
  Let $G'$ be the DAG obtained from $G \setminus Z$ by adding an edge from node $u$ to $v$, for every pair of nodes $u,v\in V(G)\setminus Z$ such that there is an odd-path from $u$ to $v$ in~$G$ all of whose internal nodes are in $Z$.
  We obtain $\torso(G, V^\infty, Z)$ from $(G', V'^\infty)$ by including a new node $x_{uv}$ and edges $u \rightarrow x_{uv}$ and $x_{uv} \rightarrow v$ for every pair of nodes $u,v \in V(G) \setminus Z$ such that there is an even path from $u$ to $v$ in~$G$ all of whose internal nodes are in $Z$.
  The set $V'^\infty$ is defined to be the union of $V^{\infty}\setminus Z$ and all the new nodes $x_{uv}$ (see Fig.~\ref{fig:torso-dag-example}).
\end{definition}

\begin{figure}[ht]
  \centering
  \begin{subfigure}[t]{0.49 \textwidth}
    \centering
    \setlength{\nodedistance}{1.2cm}
\setlength{\longnodedistance}{3.15cm}
\begin{tikzpicture}[arrows={-latex'}, node distance=\nodedistance]
	\node[vertex, label=below:$v_1$](v1) {};
	\node[vertex, label=below:$z_1$, right=of v1](z1) {};
	\node[vertex, label=below:$z_4$, below right=of z1](z4) {};
	\node[vertex, label=below:$z_2$, above right=of z4](z2) {};
	\node[vertex, label=below:$v_2$, right=of z2](v2) {};
	\node[vertex, label=below:$z_3$, below left=of z4](z3) {};
	\node[vertex, label=below:$z_5$, below right=of z4](z5) {};
	\node[vertex, label=below:$v_3$, left=of z3](v3) {};
	\node[vertex, label=below:$v_4$, right=of z5](v4) {};
	\node[vertex, label=below:$v_5$, right=of v4](v5) {};
	
	\node [draw, inner sep=0.50cm, dashed, ultra thick, fit= (z1) (z2) (z3) (z4), label=above:$Z$] {};

	\draw[very thick] (v1) -- (z1);
	\draw[very thick] (z1) -- (z2);
	\draw[very thick] (z2) -- (v2);
	\draw[very thick] (v3) -- (z2);
	\draw[very thick] (v3) -- (z3);
	\draw[very thick] (z3) -- (z4);
	\draw[very thick] (z4) -- (z5);
	\draw[very thick] (z3) -- (z5);
	\draw[very thick] (z5) -- (v4);
	\draw[very thick] (v4) -- (v5);
\end{tikzpicture}
 
    \caption{The original graph $G$.}
    \label{fig:example DAG for torso}
  \end{subfigure}
  \begin{subfigure}[t]{0.49 \textwidth}
      \centering
      \setlength{\nodedistance}{1.2cm}
\setlength{\longnodedistance}{3.15cm}
\begin{tikzpicture}[arrows={-latex'}, node distance=\nodedistance]
    \node[vertex, label=below:$x_{v_3v_4}$](x34) {};
    \node[vertex, label=left:$v_3$, above left=of x34](v3) {};
    \node[vertex, label=above:$v_4$, above right=of x34](v4) {};
    \node[vertex, label=right:$v_5$, right=of v4](v5) {};
    \node[vertex, label=right:$x_{v_3v_2}$, above right=of v3](x32) {};
	\node[vertex, label=left:$v_1$, above left=of x32](v1) {};
	\node[vertex, label=right:$v_2$, above right=of x32](v2) {};

	\draw[very thick] (v1) -- (v2);
	\draw[very thick] (v3) -- (x32);
	\draw[very thick] (x32) -- (v2);
	\draw[very thick] (v3) -- (x34);
	\draw[very thick] (x34) -- (v4);
	\draw[very thick] (v3) -- (v4);
	\draw[very thick] (v4) -- (v5);
\end{tikzpicture}
 
    \caption{$\torso(G, Z)$.}
    \label{fig:torso(G,Z)-dag}
  \end{subfigure}
  \caption{An illustration of the parity-preserving torso operation.}
\label{fig:torso-dag-example}
\end{figure}
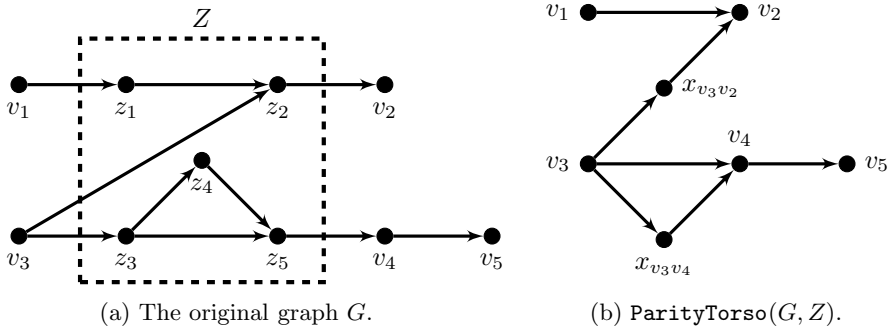

We emphasize that the acyclic nature of the input directed graph allows us to implement the parity-preserving torso operation in polynomial time (e.g., using Lemma \ref{lem:odd-path-decision}).
Moreover, applying parity-preserving torso on a DAG results in a DAG as well. 
In what follows, we state the properties of the $\torso$ operation that are exploited by our algorithm. 
The parity-preserving torso operation, has the property that it maintains $u\rightarrow v$ paths along with their parities between any pair of nodes $u,v\in V(G) \setminus Z$.
More precisely:
\begin{lemma}
\label{lem:torso does not change connectivity-dag}
  Let $G$ be a DAG and $Z, V^\infty \subseteq V(G)$.
  Also define $(G', V'^\infty) := \torso(G, V^\infty, Z)$.
  Let $u,v$ be nodes in $V(G) \setminus Z$.
  There is a $u \rightarrow v$ path $P$ in $G$ if and only if there is a $u \rightarrow v$ path $Q$ of the same parity in $G'$. Moreover, the path $Q$ can be chosen so that the nodes of $P$ in~$G \setminus Z$ are the same as the nodes of $Q$ in $G \setminus Z$, i.e.\  $V(P) \cap (V(G) \setminus Z) = V(Q) \cap (V(G) \setminus Z)$.
\end{lemma}
\begin{proof}
  We prove the forward direction by induction on the length of $P$.
  For the base case of induction, consider paths of length zero in $G$ that are disjoint from $Z$.
  Such a path is not affected by the $\torso$ operation.
  Suppose that the claim holds for all paths of length less than $\ell$, for some $\ell > 0$.
  Let $P$ be a $u \rightarrow v$ path of length $\ell$ in $G$, where $u,v \in V(G) \setminus Z$.
  If all internal nodes of $P$ are in $Z$, then by definition of $\torso$, a path $Q$ of the same parity exists in $G'$ and $V(P) \cap (V(G) \setminus Z) = V(Q) \cap (V(G) \setminus Z)$.
  Otherwise, let $w \in V(G) \setminus Z$ be an internal node of $P$.
  Let $P_1$ and $P_2$ be the subpaths of $P$ from $u$ to $w$ and from $w$ to $v$.
  By induction hypothesis, there is a $u \rightarrow w$ path $Q_1$ in $G'$ of the same parity as $P_1$ where $V(P_1) \cap (V(G) \setminus Z) = V(Q_1) \cap (V(G) \setminus Z)$, and similarly, a $w \rightarrow v$ path $Q_2$ is found of the same parity as $P_2$ where $V(P_2) \cap (V(G) \setminus Z) = V(Q_2) \cap (V(G) \setminus Z)$.
  Since $G'$ is a DAG, the path $Q$ obtained by concatenating $Q_1$ and $Q_2$ has the same parity as $P$ and $V(P) \cap (V(G) \setminus Z) = V(Q) \cap (V(G) \setminus Z)$.

  Conversely, suppose $Q$ is a path $u=x_0, x_1, x_2, \ldots, x_r=v$ in $G'$ from $u$ to $v$ where $u,v \in V(G)$.
  Then for every node $x_i$ of $Q$ in $V(G') \setminus V(G)$, replace the subpath $x_{i-1}, x_i, x_{i+1}$ with the even path in $G$ that connects $x_{i-1}$ to $x_{i+1}$.
  Also, for every pair $i$ where $x_i, x_{i+1} \in V(G)$ but $(x_i, x_{i+1})$ is not an edge in $G$, replace the subpath $x_i, x_{i+1}$ of $Q$ with the odd path that connects $x_i$ to~$x_{i+1}$ in $G$.
  By construction, the resulting sequence is a path $P$ in $G$ and has the same parity as~$Q$ and $V(P) \cap (V(G) \setminus Z) = V(Q) \cap (V(G) \setminus Z)$.
\end{proof}

\begin{corollary}
\label{cor:sol_in_torso_is_a_sol_in_G}
  Let $\mathcal{I}=(G,V^{\infty},T,k)$ be an instance of \dagoddmultiwaycut\ and let $Z\subseteq V(G)\setminus T$.
  Let $(G',V'^{\infty}):=\torso(G, V^\infty, Z)$ and denote the instance $(G',V'^{\infty},T,k)$ by $\mathcal{I}'$.
  The instance $\mathcal{I}$ admits a solution $S$ of size at most $k$ that is disjoint from $Z$ if and only if the instance $\mathcal{I'}$ admits a solution of size at most $k$.
\end{corollary}
\begin{proof}
  Let $M'$ be a solution to the instance $\mathcal{I}'$ of size at most $k$.
  Since $V^\infty\setminus Z \subseteq V'^\infty$ and $M' \cap V'^\infty = \emptyset$ and $M' \cap Z=\emptyset$, we have that $M' \cap V^\infty = \emptyset$.
  By definition of the $\torso$ operation, $V(G') \setminus V(G)$ is contained in $V'^\infty$ and therefore is disjoint from $M'$. Thus, $M' \subseteq V(G)$.
  Suppose $P$ is an odd $T$-path in $G$ and is disjoint from $M'$.
  By Lemma~\ref{lem:torso does not change connectivity-dag}, there is an odd $T$-path in $G'$ that is also disjoint from $M'$, contradicting our assumption about $M'$.

  Conversely, suppose $M$ is a solution for the instance $\mathcal{I}$ of size at most $k$ that is disjoint from~$Z$.
  Suppose $P'$ is an odd $T$-path in $G'$ and is disjoint from $M$.
  By Lemma~\ref{lem:torso does not change connectivity-dag}, There is an odd $T$-path in $G$ that is also disjoint from $M$, contradicting our assumption about $M$.
\end{proof}

Corollary~\ref{cor:sol_in_torso_is_a_sol_in_G} reveals that if there exists a solution $M$ in $G$ that is disjoint from $Z$ and~$V^{\infty}$, then it also exists in the DAG obtained from $\torso(G, V^\infty, Z)$ and hence it is sufficient to search for it in $\torso(G, V^\infty, Z)$.
Therefore, we are interested in finding a set $Z$ of nodes that is disjoint from some solution of size at most $k$, and moreover, the instance $(\torso(G, V^\infty, Z),T,k)$ is an easy instance of the problem, i.e., satisfies the property mentioned in Theorem \ref{thm:dag-easy-instances}.
The following lemma shows that it is sufficient to find a set $Z$ that contains the shadow of a solution. 
\begin{lemma}
\label{lem:shadow_has_degree_one}
  Let $G$ be a DAG and $M,Z,V^\infty \subseteq V(G)$. 
  Suppose $M$ intersects every odd $T$-path in $G$ and $s_G(M) \subseteq Z \subseteq V(G) \setminus M$.
  Let $(G', V'^\infty) := \torso(G, V^\infty, Z)$.
  Then every node in $s_{G'}(M)$ has total degree at most one in $G' \setminus M$.
\end{lemma}
\begin{proof}
  We claim that $s_{G'}(M)$ is contained in $V(G') \setminus V(G)$.
  Suppose not.
  Then there is a node $v \in V(G) \setminus Z$ that is in $s_{G'}(M)$.
  Suppose $v \in r_{G'}(M)$.
  Thus, there is no path from $v$ to $T$ in $G'$ that is disjoint from $M$.
  By Lemma~\ref{lem:torso does not change connectivity-dag}, every path in $G$ from $v$ to $T$ intersects $M$.
  Therefore,~$v$ is in the shadow of $M$ in $G$ and is hence contained in $Z$.
  This is a contradiction.
  A similar contradiction arises if $v \in f_{G'}(M)$.
  Therefore, $s_{G'}(M)$ is disjoint from $V(G)$.

  Let $x \in s_{G'}(M)$.
  We observe that by definition of the $\torso$ operation, every node $x \in V(G') \setminus V(G)$ has in-degree and out-degree one. Let $u$ and $v$ be the in-neighbor and out-neighbor of $x$ in $G'$.
  Suppose $x$ has total degree two in $G' \setminus M$.
  This implies that $u, v \notin M$. Since $u$ is not in the shadow of $M$ in $G'$, there is a $T \rightarrow u$ path disjoint from $M$ in $G'$.
  Appending the $u \rightarrow x$ edge to that path, gives a $T \rightarrow x$ path in $G'$ disjoint from $M$.
  Thus, $x \notin f_{G'}(M)$.
  Similarly, $x \notin r_{G'}(M)$, because $v \notin M$.
  Thus, $x \notin s_{G'}(M)$, a contradiction.
\end{proof}

\subsection{Difficult instances}
Corollary \ref{cor:sol_in_torso_is_a_sol_in_G} and Lemma~\ref{lem:shadow_has_degree_one} show that if we find a set $Z$ such that for some solution $M$, the set $Z$ is disjoint from $M$ and contains the shadow of $M$ in $G$, then considering $\torso(G,V^\infty,Z)$ will give a new instance that satisfies the conditions of Theorem ~\ref{thm:dag-easy-instances}.
Our goal now is to obtain such a set $Z$.
We will show the following lemma.
We emphasize that the lemma holds for arbitrary digraphs. 

\begin{lemma}
\label{lem:ShadowContainer_exists}
  There is an algorithm $\ShadowContainer$ that, given an instance $(G, V^\infty,T,k)$ of \dagoddmultiwaycut, in time $2^{O(k^2)} \poly(\card{V(G)})$ returns a family $\mathcal{Z}$ of $2^{O(k^2)}\log|V(G)|$ subsets of $V(G)$, with the property that if the instance admits a solution of size at most $k$, then for some solution $M$ of size at most $k$, there exists a set $Z\in \mathcal{Z}$ that is disjoint from $M$ and contains~$s_G(M)$.
\end{lemma}

We defer the proof of Lemma~\ref{lem:ShadowContainer_exists} to first see its implications.
We now show how the procedure $\ShadowContainer$ can be used to obtain a fixed-parameter algorithm for the \dagoddmultiwaycut problem in DAGs and thus prove Theorem \ref{thm:omwnodec-dag}.

\begin{algorithm}
\caption{Minimum odd node multiway cut in DAGs}
\label{alg:even_path_blocker}
  \begin{algorithmic}[1]
    \State \textbf{Input:} A DAG $G$ with terminal set $T$, a set $V^\infty\supseteq T$ of protected nodes, and $k\in \Z_+$.
    \State \textbf{Output:} An odd node multiway cut for $(G,T)$ of size at most $k$ and disjoint from $V^\infty$, or ``no solution of size at most $k$'' if such does not exist.
	\State $\mathcal Z\leftarrow \ShadowContainer(G, T, V^\infty, k)$ \label{line:branch}
    \For{$Z \in \mathcal Z$} 
      \State $(G_1, V_1^\infty) \gets \torso(G, V^\infty, Z)$
	  \State $N \gets \mathrm{SolveEasyInstance}(G_1, V_1^\infty, T, k)$ \label{line:solve_easy_instance}
	  \If{$N$ is a solution in $G$}
	    \State \Return $N$
	  \EndIf
	\EndFor
    \State \Return ``no solution of size at most $k$''
  \end{algorithmic}
\end{algorithm}

\begin{theorem}
\label{thm:pathblocker_is_fpt}
  There exists an algorithm that, given an instance $(G, V^\infty, T, k)$ of \dagoddmultiwaycut where $G$ is a DAG, in $2^{O(k^2)}\poly(\card{V(G)})$ time either finds a solution of size at most $k$ or determines that no such solution exists.
\end{theorem}
\begin{proof}
  We use Algorithm~\ref{alg:even_path_blocker}.
  Let $(G, V^\infty,T,k)$ be an instance of \dagoddmultiwaycut, where $G$ is a DAG.
  Suppose there exists a solution of size at most $k$.
  By Lemma~\ref{lem:ShadowContainer_exists}, the procedure $\ShadowContainer(G, T, V^\infty, k)$ in Line~\ref{line:branch} returns a family $\mathcal{Z}$ of subsets of $V(G)$ with $|\mathcal{Z}|=2^{O(k^2)}\log|V(G)|$ containing a set $Z$ such that there is a solution $M$ of size at most $k$ that is disjoint from $Z$ and $Z$ contains $s_G(M)$. 
  Let $(G_1, V_1^\infty)$ be the result of applying $\torso$ operation to the set $Z$ in $G$ (i.e., the result of Step 2 in Algorithm \ref{alg:even_path_blocker}).
  By Lemma~\ref{lem:shadow_has_degree_one}, every node in $s_{G_1}(M)$ has total degree at most one in $G_1 \setminus M$.
  Therefore, by Theorem~\ref{thm:dag-easy-instances}, the set $N$ returned in Line~\ref{line:solve_easy_instance} is a solution to the instance $(G_1, V_1^\infty, T, k)$.
  By Corollary~\ref{cor:sol_in_torso_is_a_sol_in_G}, the set $N$ is also a solution to the original instance of the problem.

  If there is no solution of size at most $k$, the algorithm will not find any.
  Therefore, the algorithm is correct.
  The runtime of the algorithm is dominated by Line 2 which can be implemented to run in $2^{O(k^2)}\text{poly}(|V(G)|)$ time by Lemma~\ref{lem:ShadowContainer_exists}.
\end{proof}

To complete this proof, it remains to prove Lemma~\ref{lem:ShadowContainer_exists}.
In order to do so, we will use the following result. 
\begin{theorem}[{Chitnis~\etal~\cite[Thm. 3.18]{ChitnisEtAl2015}}]
\label{thm:random_set_by_chitnis}
  There is an algorithm that, given a digraph~$G$, a set of protected nodes $V^{\infty}\subseteq V(G)$, terminal nodes $T\subseteq V^{\infty}$ and an integer $k$, in time $2^{O(k^2)} \poly(|V(G)|)$ returns 
  a family $\mathcal{Z}$ of subsets of $V(G)\setminus V^{\infty}$ with $|\mathcal{Z}|=2^{O(k^2)}\log|V(G)|$ 
  such that for every $S,Y \subseteq V(G)$ satisfying
  \begin{enumerate}
    \item[(i)] $S$ is a thin set with $\card{S} \leq k$, and 
    \item[(ii)] for every $v \in Y$, there exists an important $v \rightarrow T$ separator contained in $S$,
  \end{enumerate}
  there is some $Z\in \mathcal{Z}$ for which $Y \subseteq Z \subseteq V(G) \setminus S$. 
\end{theorem}

To invoke Theorem~\ref{thm:random_set_by_chitnis}, we need to guarantee that there exists a solution $S$ of size at most $k$ such that $S$ is thin and its reverse shadow $Y$ in $G$ has the property that for every $v \in Y$ there is an important $v \rightarrow T$ separator contained in $S$.
Towards obtaining such a solution, we prove the following.
\begin{lemma}
\label{lem:M'improvesM}
  Let $(G, V^\infty, T, k)$ be an instance of \dagoddmultiwaycut, where $G$ is a DAG.
  Let $M$ be a solution for this instance.
  If there exists $v \in r_G(M)$ such that $M$ does not contain an important $v \rightarrow T$ separator, then there exists another solution $M'$ of size at most~$\card{M}$, such that $r_G(M) \cup f_G(M) \cup M \subseteq r_G(M') \cup f_G(M') \cup M'$, and $r_G(M) \subsetneq r_G(M')$.
\end{lemma}
\begin{proof}
  Let $M_0$ be the set of nodes $u \in M$ for which there is a $v \rightarrow u$ path in $G$ that is internally disjoint from $M$. 
  Since $v\in r_G(M)$, every $v\rightarrow T$ path intersects $M$.
  For a $v\rightarrow T$ path $P$, the first node $u\in P\cap M$ is in $M_0$.
  Hence, every $v\rightarrow T$ path intersects $M_0$. 
  Therefore, the set $M_0$ is a $v \rightarrow T$ separator in $G$.
  Therefore, it contains a minimal separator $M_1$. Since we assumed that there is no important $v \rightarrow T$ separator contained in $M$, the set $M_1$ is not an important $v \rightarrow T$ separator. 
  Suppose $M_1$ is dominated by another $v \rightarrow T$ separator and let $M_2$ be an important $v\rightarrow T$ separator that dominates $M_1$.
  Define $M'$ as $(M \setminus M_1) \cup M_2$.
  We recall that a separator is by definition, disjoint from the protected node set.
  Therefore, $M' \cap V^\infty = \emptyset$.
  We will show that $M'$ contradicts the choice of $M$. We need the following claims.
  \begin{claim}
  \label{claim:M-M'_is_contained_in_r_G(M')}
    $M \setminus M' \subseteq r_G(M')$.
  \end{claim}
  \begin{proof}[Proof of Claim~\ref{claim:M-M'_is_contained_in_r_G(M')}]
    We observe that $M \setminus M' = M_1 \setminus M_2$.
    Let $u$ be an arbitrary node in $M_1 \setminus M_2$.
    Since $u \in M_1$ and $M_1$ is a minimal $v\rightarrow T$ separator, there is a $v \rightarrow u$ path $P_1$ that is internally disjoint from $M_1$.
    Since $M_2$ dominates $M_1$, therefore, $R_{G \setminus M_1}(v) \subseteq R_{G \setminus M_2}(v)$.
    Thus, $V(P_1)\subseteq R_{G\setminus M_2}(v)$.
    Hence, $P_1$ is disjoint from $M_2$.
    Suppose $P_2$ is an arbitrary $u \rightarrow T$ path in $G$.
    Concatenation of $P_1$ and $P_2$ is a $v \rightarrow T$ path in $G$ and therefore, has to intersect $M_2$. Since $P_1$ is disjoint from $M_2$, the path $P_2$ has to intersect $M_2$.
    Hence, every $u \rightarrow T$ path in $G$ intersects~$M_2$ and in particular, intersects $M'$.
    Equivalently, $u \in r_G(M')$.
    This completes the proof of Claim~\ref{claim:M-M'_is_contained_in_r_G(M')}.
  \end{proof}

  We next show that $M'$ is a feasible solution for the problem and is no larger than $M$.
  \begin{claim}
  \label{claim:M'isasolution}
    The set $M'$ intersects every odd $T$-path in $G$ and $\card{M'} \leq \card{M}$.
  \end{claim}
  \begin{proof}[Proof of Claim~\ref{claim:M'isasolution}]
    By assumption, every odd $T$-path $P$ intersects $M$. If $P$ intersects $M \cap M'$, then it also intersects $M'$.
    If $P$ intersects $M \setminus M'$, then by Claim~\ref{claim:M-M'_is_contained_in_r_G(M')} it also intersects $M'$.
    Thus, every odd $T$-path in $G$ intersects $M'$.
    Furthermore, by definition of $M'$, we have
    \begin{equation*}
      \card{M'}    = \card{M} + (\card{M_2 \setminus M} - \card{M_1 \setminus M_2})
                \leq \card{M} + (\card{M_2} - \card{M_1})
                \leq \card{M} \enspace .
    \end{equation*}
    This completes the proof of Claim~\ref{claim:M'isasolution}.
  \end{proof}

  \begin{claim}
  \label{claim:r_G(M)iscontainedinr_G(M')}
    $r_G(M) \subseteq r_G(M')$.
  \end{claim}
  \begin{proof}[Proof of Claim~\ref{claim:r_G(M)iscontainedinr_G(M')}]
    Let $u$ be an arbitrary node in $r_G(M)$.
    The set $M$ is a $u \rightarrow T$ separator.
    Therefore, every $u \rightarrow T$ path intersects $M$.
    We need to show that every $u \rightarrow T$ path also intersects $M'$.
    Let $P$ be a $u \rightarrow T$ path.
    If $P$ intersects $M \cap M'$, then it also intersects $M'$.
    If $P$ does not intersect $M \cap M'$, then it has to intersect $M \setminus M'$.
    By Claim~\ref{claim:M-M'_is_contained_in_r_G(M')}, every $M \setminus M' \rightarrow T$ path intersects $M'$.
    Therefore, $u \in r_G(M')$.
    This completes the proof of Claim~\ref{claim:r_G(M)iscontainedinr_G(M')}.
  \end{proof}

  \begin{claim}
  \label{claim:M'isshadowmaximal}
    $r_G(M) \cup f_G(M) \cup M \subseteq r_G(M') \cup f_G(M') \cup M'$.
  \end{claim}
  \begin{proof}[Proof of Claim~\ref{claim:M'isshadowmaximal}]
    By Claim~\ref{claim:M-M'_is_contained_in_r_G(M')}, we have $M \setminus M' \subseteq r_G(M')$ and by Claim~\ref{claim:r_G(M)iscontainedinr_G(M')}, we have $r_G(M) \subseteq r_G(M')$.
    Thus, it remains to prove that $f_G(M) \subseteq r_G(M') \cup f_G(M') \cup M'$.
    Let $u$ be an arbitrary node in $f_G(M) \setminus (r_G(M') \cup f_G(M') \cup M')$.
    Since $u \notin f_G(M')$, there is a $T \rightarrow u$ path~$P_1$ in $G$ that is disjoint from $M'$.
    But $u \in f_G(M)$.
    Thus $P_1$ has to intersect $M$, particularly it has to intersect $M \setminus M'$.
    Let $P_2$ be a subpath of $P_1$ from $M \setminus M'$ to $u$.
    Since $u \notin r_G(M')$, there is a $u \rightarrow T$ path $P_3$ in $G$ that is disjoint from $M'$.
    The concatenation of $P_2$ and $P_3$ is a path from $M \setminus M'$ to $T$ that is disjoint from $M'$.
    But by Claim~\ref{claim:M-M'_is_contained_in_r_G(M')}, every $M \setminus M' \rightarrow T$ path in~$G$ must intersect $M'$.
    This contradiction shows that $f_G(M) \subseteq (r_G(M') \cup f_G(M') \cup M')$.
    This completes the proof of Claim~\ref{claim:M'isshadowmaximal}.
  \end{proof}

  \begin{claim}
  \label{claim:r_G(M)isstrictlycontainedinr_G(M)}
    $r_G(M) \subsetneq r_G(M')$.
  \end{claim}
  \begin{proof}[Proof of Claim~\ref{claim:r_G(M)isstrictlycontainedinr_G(M)}]
    By Claim~\ref{claim:r_G(M)iscontainedinr_G(M')}, $r_G(M) \subseteq r_G(M')$.
    We need to prove $r_G(M) \neq r_G(M')$.
    We recall that $M \setminus M' = M_1 \setminus M_2$.
    Since $M_2$ is an important $v\rightarrow T$ separator, it follows that the $v\rightarrow T$ separator $M_1$ is not contained in $M_2$.
    Therefore $M \setminus M'$ is non-empty.
    Furthermore, by definition of reverse shadow, $M \setminus M'$ is not contained in $r_G(M)$, but by Claim~\ref{claim:M-M'_is_contained_in_r_G(M')}, it is contained in $r_G(M')$.
    This completes the proof of Claim~\ref{claim:r_G(M)isstrictlycontainedinr_G(M)}.
  \end{proof}
  By Claim~\ref{claim:M'isasolution}, $M'$ is a solution of size not larger than $M$.
  Therefore, the set $M'$ has the properties claimed in Lemma~\ref{lem:M'improvesM}.
  This completes the proof of Lemma~\ref{lem:M'improvesM}.
\end{proof}

We recall that a set $M \subseteq V(G)$ is \emph{thin}, if every node $v \in M$ is not in $r_G(M \setminus \set{v})$.
\begin{corollary}
\label{lem:shadow-maximal-solution}
  Let $(G, V^\infty, T, k)$ be an instance of \dagoddmultiwaycut, where $G$ is a DAG.
  Let $M^*$ be an optimal solution that maximizes the size of $\card{r_G(S) \cup f_G(S) \cup S}$ among all optimal solutions $S$.
  If more than one optimal solution maximizes this quantity, choose the one with largest $\card{r_G(S)}$.
  The set $M^*$ is thin and for every node $v \in r_G(M^*)$ there is an important $v \rightarrow T$ separator in $M^*$.
\end{corollary}
\begin{proof}
  The set $M^*$ is thin.
  If not, there is a node $v \in M^*$ such that $v \in r_G(M^* \setminus \set{v})$.
  Then $M^* \setminus \set{v}$ is a solution too, contradicting the optimality of $M^*$.

  If there is a node $v \in r_G(M^*)$ for which there is no important $v \rightarrow T$ separator in $M^*$, then by Lemma~\ref{lem:M'improvesM}, there exists a solution $M'$ such that $r_G(M) \cup f_G(M) \cup M \subseteq r_G(M') \cup f_G(M') \cup M'$, and $r_G(M) \subsetneq r_G(M')$.
  This contradicts the choice of $M^*$.
  Therefore, for every node $v \in r_G(M^*)$ there is an important $v \rightarrow T$ separator in $M^*$.
\end{proof}

We will use Corollary~\ref{lem:shadow-maximal-solution} to prove Lemma~\ref{lem:ShadowContainer_exists}.
\begin{proof}[Proof of Lemma~\ref{lem:ShadowContainer_exists}]
  Let us use $\ReverseShadowContainer(G, V^\infty, k)$ to denote the algorithm from Theorem \ref{thm:random_set_by_chitnis}. 
  We will show that Algorithm~\ref{alg:ShadowContainer} generates the desired set.
  \begin{algorithm}
  \caption{\texttt{ShadowContainer}
  \label{alg:ShadowContainer}}
    \begin{algorithmic}[1]
    \State \textbf{Input:} A digraph $G$ with terminal set $T$, a set $V^\infty$ of protected nodes containing $T$, and $k\in \Z_+$.
	\State \textbf{Output:} A set $\mathcal Z$ of at most $2^{O(k^2)}\log|V(G)|$ subsets of $V(G)$ with the property that if $(G,T,V^{\infty},k)$ admits a solution of size at most $k$, then for some solution $M$ of size at most $k$, there exists a set $Z\in \mathcal{Z}$ that is disjoint from $M$ and contains~$s_G(M)$.
  \State Let $G^\text{rev}$ denote the graph obtained from $G$ by reversing the orientation of all edges
				\State $\mathcal{Z}_1 \gets \ReverseShadowContainer(G, V^\infty, k)$    \label{line:Z1}
        \For{$Z_1\in \mathcal{Z}_1$} \label{line:branch1}
					\State $\mathcal{Z}_2 \gets \ReverseShadowContainer(G^\text{rev}, V^\infty \cup Z_1, k)$  \label{line:Z2}
					\For{$Z_2\in \mathcal{Z}_2$} \label{line:branch2}
						\State $\mathcal{Z}\gets \mathcal{Z}\cup \{Z_1\cup Z_2\}$
					\EndFor
				\EndFor
        \State \Return $\mathcal{Z}$
    \end{algorithmic}
\end{algorithm}

  By Theorem~\ref{thm:random_set_by_chitnis}, the cardinality of $\mathcal{Z}$ returned by the algorithm is $2^{O(k^2)}\log|V(G)|$. 
  The runtime analysis of the algorithm follows from the runtime analysis of the procedure $\ReverseShadowContainer$ in Theorem~\ref{thm:random_set_by_chitnis}.
  To prove the correctness of this algorithm, we argue that at least one of the sets in the returned family $\mathcal{Z}$ has the desired properties.

  Suppose there exists a solution of size at most $k$ and let $M^*$ be an optimal solution that maximizes the size of $\card{r_G(S) \cup f_G(S) \cup S}$ among all optimal solutions $S$.
  If more than one solution maximizes this quantity, choose the one with largest $\card{r_G(S)}$.
  By Corollary~\ref{lem:shadow-maximal-solution}, the solution $M^*$ is thin and has the property that every node $v$ in the reverse shadow of $M^*$ has an important $v \rightarrow T$ separator contained in $M^*$.
  By Theorem~\ref{thm:random_set_by_chitnis}, the procedure $\ReverseShadowContainer(G, V^\infty, k)$ in Line~\ref{line:Z1} will return a family $\mathcal{Z}_1$ of sets containing a set $Z_1$ that is disjoint from $M^*$ and contains its reverse shadow.
  Let us fix such a $Z_1$.

  Note that $G^\text{rev}$ is a DAG on the same node set as $G$.
  What's more, any solution for the \dagoddmultiwaycut\ instance $(G^\text{rev}, V^\infty \cup Z_1, T, k)$ is also a solution for the instance $(G, V^\infty, T, k)$.
  Conversely, a solution for the instance $(G, V^\infty, T, k)$ that is disjoint from $Z_1$ is also a solution for the instance $(G^\text{rev}, V^\infty \cup Z_1, T, k)$.
  Therefore, the set $M^*$ is also an optimal solution to the instance $(G^\text{rev}, V^\infty \cup Z_1, T, k)$.
  We observe that $f_G(S) = r_{G^\text{rev}}(S)$ and $r_G(S) = f_{G^\text{rev}}(S)$ for all $S \subseteq V(G) \setminus V^\infty$. Therefore, $M^*$ maximizes the size of $r_{G^\text{rev}}(S) \cup f_{G^\text{rev}}(S) \cup S$ among all optimal solutions $S$ to $(G^\text{rev}, V^\infty \cup Z_1, T, k)$.
  We have the following claim. 

  \begin{claim}
  \label{claim:M'=M^*}
    If for an optimal solution $M'$ for the instance $(G^\text{rev}, V^\infty \cup Z_1, T, k)$ of \dagoddmultiwaycut, we have $r_{G^\text{rev}}(M^*) \cup f_{G^\text{rev}}(M^*) \cup M^* \subseteq r_{G^\text{rev}}(M') \cup f_{G^\text{rev}}(M') \cup M'$ and $r_{G^\text{rev}}(M^*) \subseteq r_{G^\text{rev}}(M')$, then $M' = M^*$.
  \end{claim}
  \begin{proof}[Proof of Claim~\ref{claim:M'=M^*}]
    As $M^*$ maximizes $\card{r_{G^\text{rev}}(S) \cup f_{G^\text{rev}}(S) \cup S}$ among all optimal solutions for the instance $(G, V^\infty, T, k)$ and as $r_{G^\text{rev}}(M^*) \cup f_{G^\text{rev}}(M^*) \cup M^* \subseteq r_{G^\text{rev}}(M') \cup f_{G^\text{rev}}(M') \cup M'$, the two sets $r_{G^\text{rev}}(M^*) \cup f_{G^\text{rev}}(M^*) \cup M^*$ and $r_{G^\text{rev}}(M') \cup f_{G^\text{rev}}(M') \cup M'$ must be equal.
    Therefore, the set $M' \setminus M^*$ is contained inside $r_{G^\text{rev}}(M^*) \cup f_{G^\text{rev}}(M^*) \cup M^*$.
    Since nodes in $f_{G^\text{rev}}(M^*)$ are protected in $G^{\text{rev}}$ by construction, the solution $M'$ cannot contain any node from $f_{G^\text{rev}}(M^*)$.
    Since $r_{G^\text{rev}}(M^*) \subseteq r_{G^\text{rev}}(M')$ and by definition of reverse shadow, $M'$ is disjoint from $r_{G^\text{rev}}(M^*)$.
    Thus, the set $M' \setminus M^*$ is disjoint from $M^*$ and $r_{G^\text{rev}}(M^*)$ and $f_{G^\text{rev}}(M^*)$, while being contained in $r_{G^\text{rev}}(M^*) \cup f_{G^\text{rev}}(M^*) \cup M^*$.
    Hence, $M' \setminus M^* = \emptyset$ or equivalently $M' \subseteq M^*$.
    Therefore, $M' = M^*$, because $\card{M'} = \card{M^*}$.
    This completes the proof of Claim~\ref{claim:M'=M^*}.
  \end{proof}

  Suppose there is a node $v \in r_{G^\text{rev}}(M^*)$ such that no important $v \rightarrow T$ separator in $G^\text{rev}$ is contained in $M^*$.
  Then by Lemma~\ref{lem:M'improvesM}, there is another optimal solution $M'$ such that $r_{G^\text{rev}}(M^*) \cup f_{G^\text{rev}}(M^*) \cup M^* \subseteq r_{G^\text{rev}}(M') \cup f_{G^\text{rev}}(M') \cup M'$ and $r_{G^\text{rev}}(M^*) \subsetneq r_{G^\text{rev}}(M')$.
  By Claim~\ref{claim:M'=M^*}, the set $M' = M^*$, which contradicts $r_{G^\text{rev}}(M^*) \subsetneq r_{G^\text{rev}}(M')$.
  This contradiction shows that for every node $v \in r_{G^\text{rev}}(M^*)$, there is an important $v \rightarrow T$ separator in $G^\text{rev}$ that is contained in $M^*$.
  Thus, by Theorem~\ref{thm:random_set_by_chitnis}, the procedure $\ReverseShadowContainer(G^\text{rev}, V^\infty \cup Z_1, k)$ from Line \ref{line:Z2} will return a family $\mathcal{Z}_2$ of sets containing a set $Z_2$ that is disjoint from $M^*$ and contains $r_{G^\text{rev}}(M^*)=f_G(M^*)$.
  Hence $Z_1 \cup Z_2$ is disjoint from $M^*$ and contains $s_G(M^*)$.
\end{proof}

\section{\diroddpathnodeblocker\ in DAGs}
\label{sec:DAG-approx-hardness}
In this section, we prove Theorem \ref{thm:diroddpathnodeblocker-approx} by showing nearly-matching hardness of approximation (Theorem \ref{thm:hardness}) and approximability results (Theorem \ref{thm:approximation}).
We also exhibit instances of DAGs for which $\mathcal{P}^{\text{odd-cover-dir}}$ is not half-integral (Theorem \ref{thm:odd-path-blocker-polyhedron-not-half-integral}). 
\subsection{Hardness of Approximation}
\label{sec:hardness}
The main result of this section is the following:
\begin{theorem}
\label{thm:hardness}
  \diroddpathnodeblocker\ in DAGs is \NP-complete, 
  and has no efficient $(2-\varepsilon)$-approximation for any $\varepsilon>0$ assuming the Unique Games Conjecture.
\end{theorem}

As a first step, we show that \diroddpathnodeblocker\ is in $\NP$.
While this is a folklore result, we present the proof for the sake of completeness.
\begin{lemma}
\label{lem:odd-path-decision}
  There exists a polynomial-time algorithm that, given a DAG $D$ and nodes $s$ and $t$ in $D$, decides whether there exists an odd-length $s \rightarrow t$ path in $D$.
\end{lemma}
\begin{proof}
  We construct a directed bipartite graph $G$ as follows.
  For each node $v \in V(D)$, introduce nodes $v_L$ and $v_R$ in $G$.
  For each edge $uv \in E(D)$, add edges $u_Lv_R$ and $u_Rv_L$ to $G$.
  We claim that there is an odd-length $s \rightarrow t$ path in $D$ if and only if there is an $s_L \rightarrow t_R$ path in $G$. 
  Since existence of an $s_L \rightarrow t_R$ path is decidable in polynomial time, this would prove the theorem.

  We now prove the above-mentioned claim.
  Suppose $s=u^0, u^1, u^2, \ldots, u^\ell=t$ is an odd-length $s \rightarrow t$ path in $D$ with intermediate nodes $u^1,\ldots, u^{\ell-1}$.
  Then $$s_L=u_L^{0}, u_R^{1}, u_L^{2}, \ldots, u_m^{\ell}=t_m$$ is a path in $G$ and since $\ell$ is odd, we have $m=R$ and hence, the path in $G$ ends in $t_R$.
  Conversely, suppose $s_L=u_L^{0}, u_R^{1}, u_L^{2}, \ldots, u_R^{\ell}=t_R$ is a path in $G$.
  Since the path starts in one part and ends in the other, it must be of odd length.
  Therefore, $s=u^{0}, u^{1}, u^{2}, \ldots, u^{l}=t$ is an odd walk in $D$.
  Since every walk in a DAG is a path, we have an odd-length $s \rightarrow t$ path in $D$.
\end{proof}


With this result we are now ready to prove Theorem~\ref{thm:hardness}.
\begin{proof}[Proof of Theorem~\ref{thm:hardness}]
  We consider the decision version of \diroddpathnodeblocker, where the input consists of a directed acyclic graph $D$ and a non-negative integer $w$ and the goal is to decide if there exists a feasible solution for \diroddpathnodeblocker\ with at most $w$ nodes.
  Let $(D, w)$ be an instance of the decision version of \diroddpathnodeblocker.
  By Lemma~\ref{lem:odd-path-decision}, given a set of nodes $U \subseteq V(D)$ of cardinality at most $w$, we can verify in polynomial time whether $D \setminus U$ has no $s \rightarrow t$ odd-path.
  Therefore \diroddpathnodeblocker\ in DAGs is in \NP.

  To prove \NP-hardness of \diroddpathnodeblocker\ in DAGs, we give a reduction from \vc.
  Recall that the input to the \vc problem is an undirected graph $G$ and $k \in \Zbb$ and the goal is to verify if there exists a vertex cover of size at most $k$.
  We construct a directed acyclic graph $H$ from $G$ as follows: pick an arbitrary ordering of the nodes and orient the edges $\{u,v\}$ of $G$ as $u\rightarrow v$ if $u<v$ in the ordering; we add two new nodes $s,t$ with directed edges $s\rightarrow u, u\rightarrow t$ for every $u\in V(G)$.
  The resulting graph $H$ is a directed acyclic graph.
  A subset $U\subseteq V(G)$ is a vertex cover in $G$ if and only if $U$ is a feasible solution to \diroddpathnodeblocker\ in $H$.
\end{proof}

\subsection{Approximation and Integrality Gap}
\label{sec:approximation}
In this section we present an approximation algorithm of factor $2$ for \diroddpathedgeblocker\ in DAGs.
This factor matches the lower bound on the hardness of approximation shown in Section~\ref{sec:hardness}.
We will use the following integer program formulation of \diroddpathedgeblocker\ and its LP-relaxation.
\begin{alignat}{2}
	& \text{min} & & \sum_{e \in E(D)} c(e)x_e \label{eq:lp-primal}\\
	& \text{subject to}& \quad & \sum_{\mathclap{{e \in P}}}
	\begin{aligned}[t]
		x_e & \geq 1& \text{for all odd-length $s \rightarrow t$ path $P$ in $D$},\\[3ex]
		x_e & \in \set{0,1} & \quad \forall\ e \in E(D).
	\end{aligned}\nonumber
\end{alignat}
where each binary variable $x_e$ indicates whether $e$ is in the solution.
This integer program can then be relaxed to a linear program by replacing the constraints $x_e \in \set{0,1}$ with $x_e\geq0$.
We denote the resulting LP as \emph{odd path blocker LP}.
\begin{theorem}
\label{thm:approximation}
  There exists an efficient $2$-approximation algorithm for \diroddpathedgeblocker\ in DAGs.
\end{theorem}
\begin{proof}
  Our algorithm uses a construction similar to what was described in the proof of Lemma~\ref{lem:odd-path-decision}. Let $D$ be the input instance of problem \diroddpathedgeblocker\ with edge costs $c:E\rightarrow \R_+$. 
  Construct the directed bipartite graph $G$ with $V(G):=\{v_L,v_R:v\in V(D)\}$ and $E(G) := \set{u_Lv_R, u_Rv_L \suchthat uv \in E(D)}$.
  Define the cost of the edges $u_Rv_L$ and $u_Lv_R$ to be $c(uv)$.
  Let $X \subseteq E(G)$ be a minimum $s_L \rightarrow t_R$ cut in $G$.
  Let $F:=\set{uv \suchthat u_Lv_R \in X \text{ or } u_Rv_L \in X}$ be the projection of the edges of $X$ onto the edges of $D$. Claims~\ref{lem:feasibility} and \ref{lem:factor-bound} below prove that $F$ is a $2$-approximate solution for \diroddpathedgeblocker\ in $D$.
\end{proof}

\begin{claim}
\label{lem:feasibility}
  The set $F$ is a feasible solution for \diroddpathedgeblocker\ in $D$.
\end{claim}
\begin{proof}
  For sake of contradiction, suppose not.
  Then there must exist an odd $s \rightarrow t$ path in $D \setminus F$.
  Let $s=u^{0}, u^{1}, \ldots, u^{\ell}=t$ be such a path.
  We note that the edge $u^{i}u^{i+1}$ is not in $F$ for $i=0, 1, \ldots, \ell-1$.
  Thus, $u^{i}_Lu^{i+1}_R$ and $u^{i}_Ru^{i+1}_L$ are not in $X$.
  Hence, $s_L=u_L^{0}, u_R^{1}, u_L^{2}, \ldots, u_R^{\ell}=t_R$ is a path in $G$.
  This contradicts the feasibility of $X$ as an $s_L \rightarrow t_R$ cut.
  Therefore $F$ is a feasible solution.
\end{proof}

Let $x^*$ be an optimal solution to the odd path blocker LP for $D$. Let $c(x)$ denote the objective value of a feasible solution $x$ to the odd path blocker LP for $D$.
We use the same notation to denote the cost of an $s_L \rightarrow t_R$ cut in $G$.
\begin{claim}
\label{lem:factor-bound}
  The cost of $F$ is at most twice that of $x^*$.
\end{claim}
\begin{proof}
  We note that $c(F) \leq c(X)$, by the construction of $F$. It would suffice to show that $c(X) \leq 2 c(x^*)$.

  Define $Y: E(G) \rightarrow \Rbb_+$ as $Y(u_Lv_R) = Y(u_Rv_L) = x^*(uv)$ and let $c(Y):=\sum_{e\in E(G)} c_e Y(e)$.
  We have that $c(Y) = 2 c(x^*)$.
  We recall that any minimum $s_L\rightarrow t_R$ cut has the same value as an optimal solution to the following path blocking integer program, as well as its linear programming relaxation:
  \begin{align*}
                \min & \sum_{e\in E(G)}w(e)y_e\\
	\sum_{e\in P}y_e &\ge 1\ \forall\ s_L\rightarrow t_R\text{ path $P$ in $G$},\\
		         y_e &\in \{0,1\}\ \forall\ e\in E(G).
  \end{align*}
		
  Hence, it suffices to prove that $Y$ is a feasible solution to the LP-relaxation of the above integer program.
  For sake of contradiction, suppose it is not.
  It means that there is an $s_L \rightarrow t_R$ path $s_L=u_L^{0}, u_R^{1}, u_L^{2}, \ldots, u_R^{\ell}=t_R$ in $G$, such that the sum of the $Y$ values over its edges is less than one. Since $s_L$ and $t_R$ are in different parts of $G$, the length $\ell$ of this path must be odd.
  Now consider the path $s=u^{0}, u^{1}, \ldots, u^{\ell}=t$ in $D$.
  The sum of $x^*$ values on its edges is also less than one and since $\ell$ is odd, this contradicts the feasibility of $x^*$ to the odd path blocker LP for~$D$. 
  Therefore $Y$ must be feasible.
\end{proof}

The proof of Theorem~\ref{thm:approximation} also yields the following corollary.
\begin{corollary}
  The integrality gap of the odd path blocker LP in DAGs is at most $2$.
\end{corollary}

\subsection{Extreme Point Structure of the Odd Path Cover Polyhedron}
\label{sec:half-integrality}
In this section, we examine the extreme point structure of the polyhedron $\mathcal{P}^{\text{odd-cover-dir}}$ defined in Section \ref{sec:intro} but for the case of DAGs. 
Concretely, $\mathcal{P}^{\text{odd-cover-dir}}$ in a DAG is the 
set of feasible solutions to the odd path blocker LP defined in Section \ref{sec:approximation} and is given in the statement of Theorem \ref{thm:odd-path-blocker-polyhedron-not-half-integral}. 
We say that a polyhedron is \emph{half-integral} if each of its extreme points is a half-integral vector (i.e., each coordinate is an integer multiple of $1/2$). Half-integrality is a desirable property in polyhedra associated with covering LPs, because it yields a simple rounding scheme that achieves an approximation factor of $2$.
Schrijver and Seymour \cite{SchrijverSeymour1994} showed that $\mathcal{P}^{\text{odd-cover}}$ in undirected graphs is half-integral.
In this section, we exhibit a DAG for which $\mathcal{P}^{\text{odd-cover-dir}}$ has a non-half-integral extreme point. 

\begin{proof}[Proof of Theorem \ref{thm:odd-path-blocker-polyhedron-not-half-integral}]
  Consider the DAG in Fig.~\ref{fig:Escher-wall}.
  Subdivide every edge, except for the five thick edges in the top row, into two.
  All edges have unit cost.
  Let us denote the resulting DAG as $D=(V,E)$. 

  We first observe that in $D$, every odd-length path from $s$ to $t$ must use an odd number of the thick edges.
  Let $m$ be the number of edges in this network.
  We introduce a solution $x$ for this instance that is not half-integral.
  Set
  \begin{align*}
    &x(A)=1/4, x(B)=1/2, x(C)=1/4, x(D)=1/2,
    \\&
    x(E)=1/4, x(F)=1/4, x(G)=3/4, x(H)=1/4,
  \end{align*}
  and let $x(e)$ be zero for every other edge $e$.
  It can be verified that $x\in \mathcal{P}^{\text{odd-cover-dir}}$ for this instance.
  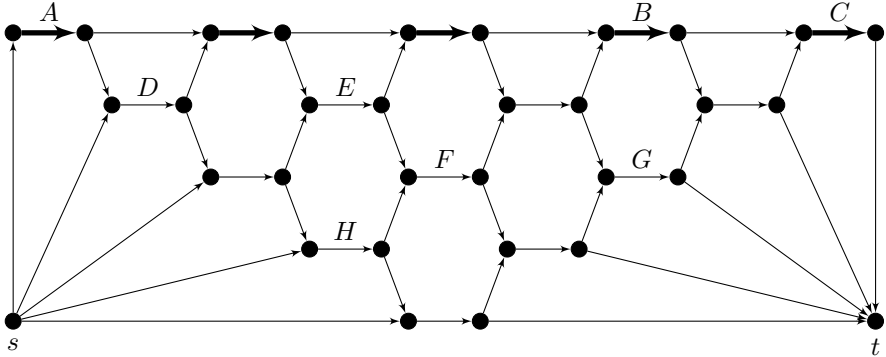
\begin{figure}[ht]
    \centering
    \setlength{\horiznodedistance}{0.2cm}
\setlength{\vertnodedistance}{0.8cm}
\setlength{\nodedistance}{0.73cm}
\setlength{\longnodedistance}{5.03cm}
\begin{tikzpicture}[arrows={-latex'}, node distance=\nodedistance]
    \node[vertex](v30) {};
    \node[vertex, right=of v30](v31) {};
    \node[vertex, above left=\vertnodedistance and \horiznodedistance of v30](v26) {};
    \node[vertex, left=of v26](v25) {};
    \node[vertex, above right=\vertnodedistance and \horiznodedistance of v31](v27) {};
    \node[vertex, right=of v27](v28) {};
    \node[vertex, above left=\vertnodedistance and \horiznodedistance of v25](v20) {};
    \node[vertex, left=of v20](v19) {};
    \node[vertex, above right=\vertnodedistance and \horiznodedistance of v26](v21) {};
    \node[vertex, above left=\vertnodedistance and \horiznodedistance of v27](v22) {};
    \node[vertex, above right=\vertnodedistance and \horiznodedistance of v28](v23) {};
    \node[vertex, right=of v23](v24) {};
    \node[vertex, above left=\vertnodedistance and \horiznodedistance of v19](v12) {};
    \node[vertex, left=of v12](v11) {};
    \node[vertex, above right=\vertnodedistance and \horiznodedistance of v20](v13) {};
    \node[vertex, above left=\vertnodedistance and \horiznodedistance of v21](v14) {};
    \node[vertex, above right=\vertnodedistance and \horiznodedistance of v22](v15) {};
    \node[vertex, above left=\vertnodedistance and \horiznodedistance of v23](v16) {};
    \node[vertex, above right=\vertnodedistance and \horiznodedistance of v24](v17) {};
    \node[vertex, right=of v17](v18) {};
    \node[vertex, above left=\vertnodedistance and \horiznodedistance of v11](v2) {};
    \node[vertex, left=of v2](v1) {};
    \node[vertex, above right=\vertnodedistance and \horiznodedistance of v12](v3) {};
    \node[vertex, above left=\vertnodedistance and \horiznodedistance of v13](v4) {};
    \node[vertex, above right=\vertnodedistance and \horiznodedistance of v14](v5) {};
    \node[vertex, above left=\vertnodedistance and \horiznodedistance of v15](v6) {};
    \node[vertex, above right=\vertnodedistance and \horiznodedistance of v16](v7) {};
    \node[vertex, above left=\vertnodedistance and \horiznodedistance of v17](v8) {};
    \node[vertex, above right=\vertnodedistance and \horiznodedistance of v18](v9) {};
    \node[vertex, right=of v9](v10) {};
    \node[vertex, left=\longnodedistance of v30, label=below:{$s$}](v29) {}; 
    \node[vertex, right=\longnodedistance of v31, label=below:{$t$}](v32) {};
    
    \draw[thickedge] (v1) -- (v2) node [midway, above, fill=none] {$A$};
    \draw[thinedge] (v2) -- (v3);
    \draw[thickedge] (v3) -- (v4);
    \draw[thinedge] (v4) -- (v5);
    \draw[thickedge] (v5) -- (v6);
    \draw[thinedge] (v6) -- (v7);
    \draw[thickedge] (v7) -- (v8) node [midway, above, fill=none] {$B$};
    \draw[thinedge] (v8) -- (v9);
    \draw[thickedge] (v9) -- (v10) node [midway, above, fill=none] {$C$};
    \draw[thinedge] (v11) -- (v12) node [midway, above, fill=none] {$D$};
    \draw[thinedge] (v13) -- (v14) node [midway, above, fill=none] {$E$};
    \draw[thinedge] (v15) -- (v16);
    \draw[thinedge] (v17) -- (v18);
    \draw[thinedge] (v19) -- (v20);
    \draw[thinedge] (v21) -- (v22) node [midway, above, fill=none] {$F$};
    \draw[thinedge] (v23) -- (v24) node [midway, above, fill=none] {$G$};
    \draw[thinedge] (v25) -- (v26) node [midway, above, fill=none] {$H$};
    \draw[thinedge] (v27) -- (v28);
    \draw[thinedge] (v30) -- (v31);
    \draw[thinedge] (v2) -- (v11);
    \draw[thinedge] (v12) -- (v19);
    \draw[thinedge] (v20) -- (v25);
    \draw[thinedge] (v26) -- (v30);
    \draw[thinedge] (v4) -- (v13);
    \draw[thinedge] (v14) -- (v21);
    \draw[thinedge] (v22) -- (v27);
    \draw[thinedge] (v6) -- (v15);
    \draw[thinedge] (v16) -- (v23);
    \draw[thinedge] (v8) -- (v17);
    \draw[thinedge] (v12) -- (v3);
    \draw[thinedge] (v20) -- (v13);
    \draw[thinedge] (v14) -- (v5);
    \draw[thinedge] (v26) -- (v21);
    \draw[thinedge] (v22) -- (v15);
    \draw[thinedge] (v16) -- (v7);
    \draw[thinedge] (v31) -- (v27);
    \draw[thinedge] (v28) -- (v23);
    \draw[thinedge] (v24) -- (v17);
    \draw[thinedge] (v18) -- (v9);
    \draw[thinedge] (v29) -- (v1);
    \draw[thinedge] (v29) -- (v11);
    \draw[thinedge] (v29) -- (v19);
    \draw[thinedge] (v29) -- (v25);
    \draw[thinedge] (v29) -- (v30);
    \draw[thinedge] (v31) -- (v32);
    \draw[thinedge] (v28) -- (v32);
    \draw[thinedge] (v24) -- (v32);
    \draw[thinedge] (v18) -- (v32);
    \draw[thinedge] (v10) -- (v32);
\end{tikzpicture}
    \caption{An instance of \diroddpathedgeblocker. All edges are subdivided into two, except for the five thick edges in the top row.}
    \label{fig:Escher-wall}
  \end{figure}

  Next, we show that this solution $x$ is an extreme point of $\mathcal{P}^{\text{odd-cover-dir}}$, i.e., $x$ is an optimal solution to the odd path blocker LP.
  For this, we find a solution to the dual linear program with the same objective value.
  Let $\mathcal{Q}_{s \rightarrow t}$ be the collection of edge-sets corresponding to odd-length paths from $s$ to $t$.
  As the dual of odd path blocker LP, we obtain:
  \begin{alignat}{2}
	& \text{max} & & \sum_{P \in \mathcal{Q}_{s\rightarrow t}} f_p \label{eq:lp-dual}\\
	& \text{subject to}& \quad & \sum_{\mathclap{{P \in \mathcal{Q}_{s\rightarrow t}:e \in P}}}
	\begin{aligned}[t]
		f_P & \leq c(e),&& \text{for each edge $e$}\\[3ex]
		f_P & \geq 0, && \text{for all}~P\in \mathcal{Q}_{s \rightarrow t}
	\end{aligned}\nonumber
  \end{alignat}
  Let us call the dual LP as odd flow packing LP.
  The dual formulation describes the problem of sending the maximum flow along odd paths in the network, such that the amount of flow going through each edge does not exceed its capacity. 

  Consider the following paths in Fig.~\ref{fig:Escher-flow}:
  \begin{align*}
    P_1= &(29, 1, 2, 11, 12, 19, 20, 25, 26, 30, 31, 32)\\
    P_2= &(29, 11, 12, 3, 4, 13, 14, 21, 22, 27, 28, 32)\\
    P_3= &(29, 19, 20, 13, 14, 5, 6, 15, 16, 23, 24, 32)\\
    P_4= &(29, 25, 26, 21, 22, 15, 16, 7, 8, 17, 18, 32)\\
    P_5= &(29, 30, 31, 27, 28, 23, 24, 17, 18, 9, 10, 32)\\
    P_6= &(29, 1, 2, 3, 4, 5, 6, 7, 8, 9, 10, 32)
  \end{align*}
  The dual solution that we introduce sends a flow of value $1/2$ along each of $P_1, P_2, \ldots, P_6$.
  The total flow, therefore would be $3$.
  By strong duality condition, the feasible solutions of an LP and its dual match only at optimal solutions.
  Therefore the primal solution $x$ is optimal.
  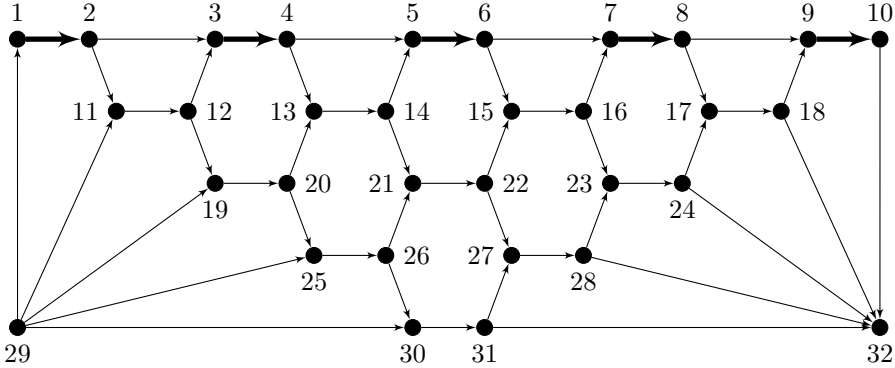
\begin{figure}[ht]
    \centering
    \setlength{\horiznodedistance}{0.2cm}
\setlength{\vertnodedistance}{0.8cm}
\setlength{\nodedistance}{0.73cm}
\setlength{\longnodedistance}{5.03cm}
\begin{tikzpicture}[arrows={-latex'}, node distance=\nodedistance]
    \node[vertex, label=below:{$30$}](v30) {};
    \node[vertex, right=of v30, label=below:{$31$}](v31) {};
    \node[vertex, above left=\vertnodedistance and \horiznodedistance of v30, label=right:{$26$}](v26) {};
    \node[vertex, left=of v26, label=below:{$25$}](v25) {};
    \node[vertex, above right=\vertnodedistance and \horiznodedistance of v31, label=left:{$27$}](v27) {};
    \node[vertex, right=of v27, label=below:{$28$}](v28) {};
    \node[vertex, above left=\vertnodedistance and \horiznodedistance of v25, label=right:{$20$}](v20) {};
    \node[vertex, left=of v20, label=below:{$19$}](v19) {};
    \node[vertex, above right=\vertnodedistance and \horiznodedistance of v26, label=left:{$21$}](v21) {};
    \node[vertex, above left=\vertnodedistance and \horiznodedistance of v27, label=right:{$22$}](v22) {};
    \node[vertex, above right=\vertnodedistance and \horiznodedistance of v28, label=left:{$23$}](v23) {};
    \node[vertex, right=of v23, label=below:{$24$}](v24) {};
    \node[vertex, above left=\vertnodedistance and \horiznodedistance of v19, label=right:{$12$}](v12) {};
    \node[vertex, left=of v12, label=left:{$11$}](v11) {};
    \node[vertex, above right=\vertnodedistance and \horiznodedistance of v20, label=left:{$13$}](v13) {};
    \node[vertex, above left=\vertnodedistance and \horiznodedistance of v21, label=right:{$14$}](v14) {};
    \node[vertex, above right=\vertnodedistance and \horiznodedistance of v22, label=left:{$15$}](v15) {};
    \node[vertex, above left=\vertnodedistance and \horiznodedistance of v23, label=right:{$16$}](v16) {};
    \node[vertex, above right=\vertnodedistance and \horiznodedistance of v24, label=left:{$17$}](v17) {};
    \node[vertex, right=of v17, label=right:{$18$}](v18) {};
    \node[vertex, above left=\vertnodedistance and \horiznodedistance of v11, label=above:{$2$}](v2) {};
    \node[vertex, left=of v2, label=above:{$1$}](v1) {};
    \node[vertex, above right=\vertnodedistance and \horiznodedistance of v12, label=above:{$3$}](v3) {};
    \node[vertex, above left=\vertnodedistance and \horiznodedistance of v13, label=above:{$4$}](v4) {};
    \node[vertex, above right=\vertnodedistance and \horiznodedistance of v14, label=above:{$5$}](v5) {};
    \node[vertex, above left=\vertnodedistance and \horiznodedistance of v15, label=above:{$6$}](v6) {};
    \node[vertex, above right=\vertnodedistance and \horiznodedistance of v16, label=above:{$7$}](v7) {};
    \node[vertex, above left=\vertnodedistance and \horiznodedistance of v17, label=above:{$8$}](v8) {};
    \node[vertex, above right=\vertnodedistance and \horiznodedistance of v18, label=above:{$9$}](v9) {};
    \node[vertex, right=of v9, label=above:{$10$}](v10) {};
    \node[vertex, left=\longnodedistance of v30, label=below:{$29$}](v29) {}; 
    \node[vertex, right=\longnodedistance of v31, label=below:{$32$}](v32) {};
    
    \draw[thickedge] (v1) -- (v2);
    \draw[thinedge] (v2) -- (v3);
    \draw[thickedge] (v3) -- (v4);
    \draw[thinedge] (v4) -- (v5);
    \draw[thickedge] (v5) -- (v6);
    \draw[thinedge] (v6) -- (v7);
    \draw[thickedge] (v7) -- (v8);
    \draw[thinedge] (v8) -- (v9);
    \draw[thickedge] (v9) -- (v10);
    \draw[thinedge] (v11) -- (v12);
    \draw[thinedge] (v13) -- (v14);
    \draw[thinedge] (v15) -- (v16);
    \draw[thinedge] (v17) -- (v18);
    \draw[thinedge] (v19) -- (v20);
    \draw[thinedge] (v21) -- (v22);
    \draw[thinedge] (v23) -- (v24);
    \draw[thinedge] (v25) -- (v26);
    \draw[thinedge] (v27) -- (v28);
    \draw[thinedge] (v30) -- (v31);
    \draw[thinedge] (v2) -- (v11);
    \draw[thinedge] (v12) -- (v19);
    \draw[thinedge] (v20) -- (v25);
    \draw[thinedge] (v26) -- (v30);
    \draw[thinedge] (v4) -- (v13);
    \draw[thinedge] (v14) -- (v21);
    \draw[thinedge] (v22) -- (v27);
    \draw[thinedge] (v6) -- (v15);
    \draw[thinedge] (v16) -- (v23);
    \draw[thinedge] (v8) -- (v17);
    \draw[thinedge] (v12) -- (v3);
    \draw[thinedge] (v20) -- (v13);
    \draw[thinedge] (v14) -- (v5);
    \draw[thinedge] (v26) -- (v21);
    \draw[thinedge] (v22) -- (v15);
    \draw[thinedge] (v16) -- (v7);
    \draw[thinedge] (v31) -- (v27);
    \draw[thinedge] (v28) -- (v23);
    \draw[thinedge] (v24) -- (v17);
    \draw[thinedge] (v18) -- (v9);
    \draw[thinedge] (v29) -- (v1);
    \draw[thinedge] (v29) -- (v11);
    \draw[thinedge] (v29) -- (v19);
    \draw[thinedge] (v29) -- (v25);
    \draw[thinedge] (v29) -- (v30);
    \draw[thinedge] (v31) -- (v32);
    \draw[thinedge] (v28) -- (v32);
    \draw[thinedge] (v24) -- (v32);
    \draw[thinedge] (v18) -- (v32);
    \draw[thinedge] (v10) -- (v32);
\end{tikzpicture}
    \caption{In this dual solution, flows of value $1/2$ are being sent through $p_1, p_2, \ldots, p_6$.}
    \label{fig:Escher-flow}
  \end{figure}

  Finally, we prove that the primal solution $x$ is an extreme point of the polyhedron $\mathcal{P}^{\text{odd-cover-dir}}$ for this instance.
  For this, we present $m$ linearly independent constraints of the odd path blocker LP that are satisfied as equations by $x$.
  For each of the $m-8$ edges that are not in the support of~$x$, the non-negativity constraint is tight. 
  Consider $P_1, P_2, \ldots, P_6$ above, along with the following two odd-length $s\rightarrow t$ paths in $D$:
  \begin{align*}
    P_7= &(29, 1, 2, 3, 4, 5, 6, 15, 16, 23, 24, 32)\\
    P_8= &(29, 11, 12, 3, 4, 5, 6, 7, 8, 17, 18, 32)
  \end{align*}
  We observe that the constraints corresponding to each of these paths are tight with respect to $x$.
  It remains to prove that these $m$ constraints are linearly independent.
  Since each of the non-negativity constraints have exactly one non-zero entry and no two of them have the same non-zero entry, they are linearly independent.
  To prove that all the constraints are linearly independent, it remains to show that the path constraints are linearly independent when restricted to the edges not in the support of $x$.
  This can be verified by computing the determinant of the corresponding constraint matrix. 
\end{proof}

\section{\oddpathedgeblocker\ in undirected graphs}
\label{sec:undir-approximability}
In this section, we focus on the approximability of \oddpathedgeblocker.
We will show inapproximability results and an integrality gap instance which suggests that new techniques might be needed to improve on the approximation factor. 




\subsection{Hardness of Approximation}
\label{sec:undir-np-hardness}
In this section we prove Theorem \ref{thm:odd-path-edge-blocker-hardness}, i.e., \NP-hardness of \undirpathblocker.
\undirOddPathBlockerHardness*
\begin{proof}
  Edmonds gave a polynomial time algorithm to decide whether a given undirected graph with nodes $s$ and $t$ has an odd-length $s-t$ path; cf. LaPaugh and Papadimitriou~\cite{LaPaughPapadimitriou1984}. 
  Therefore, given a candidate solution, one can verify its feasibility in polynomial time.
  Thus, \undirpathblocker is in \NP.
  We will show that the decision version of \undirpathblocker is \NP-complete by a polynomial-time reduction from \mwc.
  We recall that the input to \mwc is an undirected graph $G$, a collection $T$ of nodes in $G$ known as terminals and $k\in \Z_+$, and the goal is to verify if there exists a subset of at most $k$ edges of $G$ whose deletion ensures that no pair of terminals can reach each other. 

  Suppose $(G, T,k)$ is an instance of \mwc where $G$ is an undirected graph with $n$ nodes and $m$ edges, and $T \subseteq V(G)$ is the set of terminals to be separated.
  We obtain a graph $H$ from $G$ as follows: Let $V(H) := V(G) \cup \set{s,t} \cup \set{x_v, x'_v \suchthat v \in T}$, where $s,t, x_v$ and $x'_v$ are newly introduced nodes and let $E(H) := E(G) \cup \set{\set{x_v,v}, \set{x'_v, v}, \set{x_v, x'_v}, \set{s,x_v}, \set{t,x_v} \suchthat v \in T}$; replace every edge in $E(H)\setminus E(G)$ by $m+1$ parallel edges (see Fig.~\ref{fig:undir-hardness-gadget}). 
  For $k\le m$, we claim that the \mwc instance $(G, T,k)$ has a solution of size at most $k$ if and only if the \undirpathblocker instance $H$ has a solution of size at most $k$.
  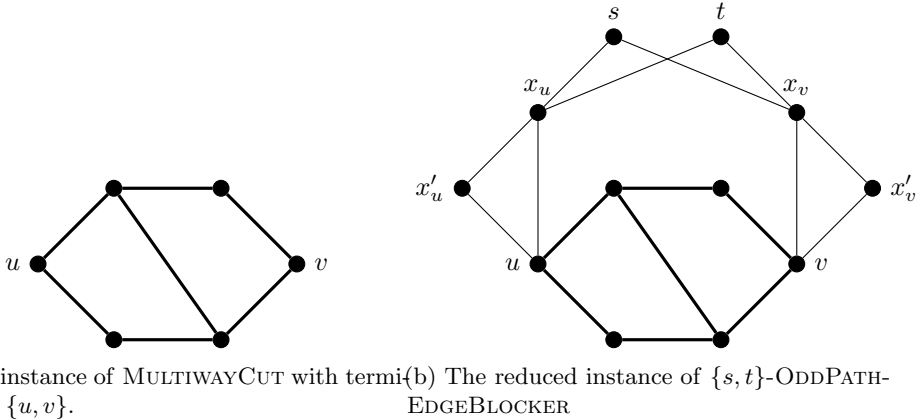
\begin{figure}[ht]
    \centering
      \begin{subfigure}[t]{0.49 \textwidth}
        \centering
        \setlength{\nodedistance}{1.2cm}
\setlength{\longnodedistance}{3.15cm}
\begin{tikzpicture}[node distance=\nodedistance]
	\node[vertex, label=left:$u$](u) {};
	\node[vertex, above right=of u](z1) {};
	\node[vertex, below right=of u](z3) {};
	\node[vertex, right=of z1](z2) {};
	\node[vertex, right=of z3](z4) {};
	\node[vertex, label=right:$v$, below right=of z2](v) {};

	\draw[very thick] (u) -- (z1);
	\draw[very thick] (z1) -- (z2);
	\draw[very thick] (z2) -- (v);
	\draw[very thick] (v) -- (z4);
	\draw[very thick] (z4) -- (z3);
	\draw[very thick] (z3) -- (u);
	\draw[very thick] (z4) -- (z1);
\end{tikzpicture}
 
        \caption{An instance of \mwc with terminal set $\set{u,v}$.}
        \label{fig:example undir graph for reduction}
      \end{subfigure}
      \begin{subfigure}[t]{0.49 \textwidth}
        \centering
        \setlength{\nodedistance}{1.2cm}
\setlength{\longnodedistance}{3.15cm}
\begin{tikzpicture}[node distance=\nodedistance]
	\node[vertex, label=left:$u$](u) {};
	\node[vertex, above right=of u](z1) {};
	\node[vertex, below right=of u](z3) {};
	\node[vertex, right=of z1](z2) {};
	\node[vertex, right=of z3](z4) {};
	\node[vertex, label=right:$v$, below right=of z2](v) {};
	
	\node[vertex, label=above:$x_u$, above left=of z1](xu) {};
	\node[vertex, label=left:$x_u'$, below left=of xu](xu') {};
	\node[vertex, label=above:$x_v$, above right=of z2](xv) {};
	\node[vertex, label=right:$x_v'$, below right=of xv](xv') {};

	\node[vertex, label=above:$s$, above right=of xu](s) {};
	\node[vertex, label=above:$t$, above left=of xv](t) {};

	\draw[very thick] (u) -- (z1);
	\draw[very thick] (z1) -- (z2);
	\draw[very thick] (z2) -- (v);
	\draw[very thick] (v) -- (z4);
	\draw[very thick] (z4) -- (z3);
	\draw[very thick] (z3) -- (u);
	\draw[very thick] (z4) -- (z1);
	
	\draw (u) -- (xu);
	\draw (u) -- (xu');
	\draw (xu) -- (xu');
	\draw (s) -- (xu);
	\draw (t) -- (xu);
	\draw (v) -- (xv);
	\draw (v) -- (xv');
	\draw (xv) -- (xv');
	\draw (s) -- (xv);
	\draw (t) -- (xv);
\end{tikzpicture}
 
        \caption{The reduced instance of \undirpathblocker}
        \label{fig:reduced instance}
      \end{subfigure}
    \caption{An illustration of the reduction from \mwc to \undirpathblocker.}
    \label{fig:undir-hardness-gadget}
  \end{figure}
  Suppose $F \subseteq E(G)$ is a solution to \mwc in $(G,T)$.
  If $F$ is not a solution to \undirpathblocker in $H$, then there is an odd path from $s$ to $t$ in $H \setminus F$.
  By construction, this path is of the form $sP_1vP_2uP_3t$, where $v,u \in T$.
  Hence, there is a $v-u$ path $P_2$ in $G \setminus F$, contradicting the feasibility of $F$ as a solution to \mwc in $(G,T)$.

  Conversely, suppose $F \subseteq E(H)$ is a solution to \undirpathblocker in $H$ of size at most $k$.
  Since $k\le m$, and by construction, we may assume that $F\subseteq E(G)$. 
  Suppose $F$ is not a solution to \mwc in $(G,T)$.
  Then there is a $v-u$ path $P$ in $G \setminus F$ for some distinct $u,v \in T$.
  If $P$ is even, then let $Q$ be the path $sx_vvPux_{u}'x_ut$ in $G'$ and if $P$ is odd, then let $Q$ be the path $sx_vvPux_ut$ in $G'$.
  In both cases, the path $Q$ is an odd-length $s-t$ in $H-F$, contradicting the feasibility of $F$ as a solution to \undirpathblocker in $(G, E^\infty, s, t)$.

  We note that the above reduction is an approximation factor preserving reduction.
  It is known that \mwc is \NP-hard and does not admit a polynomial-time approximation scheme, unless $\mathsf{P} = \NP$ \cite{DahlhausEtAl1994}.
  Moreover, there is no efficient $(6/5-\varepsilon)$-approximation for \mwc assuming the Unique Games Conjecture \cite{ManokaranEtAl2008,AngelidakisEtAl2017}.
  Hence, the results follow.
\end{proof}


\subsection{Integrality Gap}
\label{sec:undir-gap}
By the half-integrality of the extreme points of $P^{\text{odd-cover}}$ \cite{SchrijverSeymour1994}, we have a $2$-approximation algorithm for \undirpathblocker by solving the LP-relaxation of the odd path blocker LP.
The following proposition shows that the integrality gap of the odd path blocker LP is indeed $2$ and hence we cannot hope to improve on the $2$-approximation using the odd path blocker LP.
\begin{lemma}
\label{lem:undir-odd-path-blocker-lp-integrality-gap}
  The integrality gap of the following odd path blocker LP for \undirpathblocker is at least $2$:
  \begin{align*}
                \min &\sum_{e\in E} c(e)x_e &\\
    &\sum_{e\in P}x_e \ge 1, &\textnormal{for all odd-length $s-t$ paths $P$ in $G$}\\
     &            x_e \ge 0, &e\in E(G).
  \end{align*}
\end{lemma}
\begin{proof}
  For every $k \in \Z_+$, we construct a graph $G$ for which the integrality gap of the odd path blocker LP is at least $2(1-1/k)$.
  Let $S_k$ be the star graph on $k$ nodes (i.e., $V(S_k):=\{u_1,u_2,\ldots,u_k\}$ and $E(S_k):=\{u_iu_k:i\in\{1,\ldots,k-1\}\}$).
  Let $T$ be the set of leaves of $S_k$.
  Let~$G$ be the graph obtained from $(S_k,T)$ by applying the construction in the proof of Theorem~\ref{thm:odd-path-edge-blocker-hardness}.
  An optimal solution to the odd path blocker LP for $G$, assigns the value $1/2$ to every edge of $S_k$, while an optimal integral solution removes all but one edges of $S_k$ that are in $G$.
  Therefore, the ratio of the integral solution to the fractional solution is $(k-1)/(k/2)=2(1-1/k)$.
\end{proof}
The same arguments as in Lemma \ref{lem:undir-odd-path-blocker-lp-integrality-gap} also show that the integrality gap of the odd path node blocker LP for \undirpathnodeblocker is also at least $2$.

\section{Discussion}
\label{sec:discussion}
In this work, we studied a natural cut problem with parity constraints in undirected graphs and directed graphs.
Our main results are fixed-parameter algorithms parameterized by the solution size, as well as constant-factor polynomial-time approximation algorithms and inapproximability reductions.

Several questions lend themselves for future work.
Firstly, it would be interesting to determine the exact approximability bound achievable in polynomial time for \omwedgec\ in undirected graphs, closing the gap between the lower bound of $6/5$ and the upper bound of 2.

Secondly, an important line of investigation in parameterized complexity is the design of fixed-parameter algorithms that have the best possible asymptotic dependence on the parameters (modulo the Exponential-Time Hypothesis), with only linear time dependence on the instance size~\cite{LokshtanovEtAl2016,LokshtanovEtAl2017,EtscheidMnich2017,RamanujanSaurabh2017,IwataEtAl2014,MarxEtAl2013}.
Thus, it would be interesting to know whether the proposed fixed-parameter algorithms for \omwnodec in DAGs can be expedited to run in time $2^{O(k)} \cdot O(n)$.

Finally, we ask about the parameterized complexity of the common generalization of \omwnodec and {\sc Multicut} known as {\sc Odd Multicut}: given a graph $G$ with terminal pairs $\{s_1,t_1\},\hdots,\{s_p,t_p\}$ and an integer $k$, decide if some set $S$ of at most $k$ nodes intersects all odd-paths between $s_i$ and $t_i$, for $i = 1,\hdots,p$.
Does this problem admit a fixed-parameter algorithm parameterized by the solution size $k$?
\bibliographystyle{siamplain}
\bibliography{refs}
\appendix
\section{Equivalence of \omwedgec and \omwnodec}
In this section, we show an approximation-preserving and parameter-preserving equivalence between \omwedgec and \omwnodec in directed graphs. 
To establish the notation, we restate the definitions of the two problems.

\defparproblem{\omwedgec}{$k$}{A directed acyclic graph $G$ with a set $T\subseteq V(G)$ of terminal nodes and a set $E^{\infty}\subseteq E(G)$ of protected edges, and an integer $k\in\Z_+$.}{Verify if there exists an odd multiway edge cut of $T$ in $G$ of size at most $k$ and disjoint from $E^\infty$, that is, a set $M \subseteq E(G) \setminus E^\infty$ of edges that intersects every odd $T$-path in $G$.}

\defparproblem{\omwnodec}{$k$}{A directed acyclic graph $G$ with a set $T\subseteq V$ of terminal nodes and a set $V^{\infty}\subseteq V(G)$ of protected nodes, and an integer $k\in\Z_+$.}{Verify if there exists an odd multiway node cut in $G$ of size at most $k$ and disjoint from $V^\infty$, that is, a set $M \subseteq V(G) \setminus V^\infty$ of nodes that intersects every odd $T$-path in $G$.}


\begin{lemma}\label{lem:omwedgec-to-nodec}
  There exist approximation-preserving and parameter-preserving reductions between \omwedgec in directed graphs and \omwnodec in directed graphs. 
\end{lemma}
\begin{proof}
  We first show an approximation preserving and parameter preserving reduction from \omwedgec\ to \omwnodec. 
  Let $\calI=(G, E^\infty, T, k)$ be an instance of the \omwedgec problem in directed graphs.
  We create a new graph~$G'$ by subdividing every edge $e \in E(G) \setminus E^\infty$ into three edges, by introducing two new nodes~$x_e$ and~$y_e$.
  By construction,~$G'$ is a directed graph too.
  Define $V^\infty := V(G)$.
  We claim that the instance~$\calI$ of \omwedgec has a solution of size, say $r$, if and only if the instance $\calI'=(G', V^\infty, T, k)$ of \omwnodec has a solution of size $r$.

  Let $M \subseteq E(G) \setminus E^\infty$ be an odd edge multiway cut in $G$.
  We claim that $M' := \set{x_e \suchthat e \in M}$ is an odd multiway cut in $G'$.
  Let $P'$ be an odd $T$-path in $G'$. Let $P$ be the corresponding path in $G$.
  By construction, $P$ has the same parity as $P'$ and therefore, is odd.
  Since $M$ is an odd multiway cut in $G$, it must contain an edge $e \in E(P)$.
  Therefore, $M'$ contains $x_e$. Thus, every odd $T$-path in $G'$ intersects $M'$.
  Hence, $M'$ is a solution for the instance $\calI'$ of the \omwnodec problem.
  Moreover, $\card{M'} = \card{M}$.

  Conversely, suppose $N'$ is a solution to the instance $\calI'$ of the \omwnodec problem.
  Define $N := \set{e \suchthat x_e \in N' \text{ or } y_e \in N'}$.
  By construction,~$N$ is disjoint from $E^\infty$.
  Let~$P$ be an odd $T$-path in~$G$ and let~$P'$ be the corresponding path in~$G'$.
  The path $P'$ is an odd $T$-path in~$G'$.
  Since~$N'$ is an odd multiway node cut in~$G'$, it must contain a node $v \in V(P')$.
  Since $V(G) \subseteq V^\infty$, it must be the case that $v \in \bigcup_{e \in E(G)} \set{x_e,y_e}$.
  Suppose $v \in \set{x_e,y_e}$ for some $e \in E(G)$. By construction, we have $e \in N$.
  Therefore,~$N$ includes an edge from every odd $T$-path in $G$.
  Thus, $N$ is a solution to the instance $\calI$ of the \omwedgec problem.
  Moreover, $\card{N} = \card{N'}$.

	We next show an approximation preserving and parameter preserving reduction from \omwnodec\ to \omwedgec. 
  Let $\calI=(G, V^\infty, T, k)$ be an instance of the \omwnodec problem in directed graphs.
  We create a new graph~$G'$ as follows. For every node $v \in V(G)$, we create three nodes $v_\text{in}$, $v_\text{mid}$ and~$v_\text{out}$ and put an edge from~$v_\text{in}$ to~$v_\text{mid}$ and an edge from $v_\text{mid}$ to $v_\text{out}$.
  For every edge $uv \in E(G)$, we put an edge from~$u_\text{out}$ to~$v_\text{in}$ in~$G'$.
  Define $E^\infty := \set{u_\text{out}v_\text{in} \suchthat uv \in E(G)} \cup \set{v_\text{in}v_\text{mid}, v_\text{mid}v_\text{out} \suchthat v \in V^\infty}$ and define~$T'$ as $\set{t_\text{in} \suchthat t \in T}$.
  We claim that the instance $\calI$ of the \omwnodec problem has a solution of size, say $r$, if and only if the instance $\calI' := (G', E^\infty, T', k)$ of \omwedgec has a solution of size $r$.

  Let $M \subseteq V(G)$ be a solution to the instance $\calI$ of the \omwedgec.
  We claim that $M' := \set{v_\text{mid}v_\text{out} \suchthat v \in M}$ is an odd multiway cut in~$G'$.
  By construction, $M'$ is disjoint from~$E^\infty$. We show that it intersects every odd $T'$-path in~$G'$.
  Let~$P'$ be an odd $T'$-path in~$G'$. Let~$P$ be the corresponding path in~$G$.
  By construction,~$P$ has the same parity as $P'$ and therefore, is odd.
  Since~$M$ is an odd multiway cut in $G$, it must contain a node $v \in V(P)$.
  Therefore,~$M'$ contains $v_\text{mid}v_\text{out}$.
  Thus, every odd $T'$-path in~$G'$ intersects~$M'$.
  Hence,~$M'$ is a solution for the instance~$\calI'$ of \omwedgec. Moreover $\card{M'} = \card{M}$.

  Conversely, suppose~$N'$ is a solution to the instance~$\calI'$ of the \omwnodec problem.
  Define $N := \set{v \suchthat v_\text{in}v_\text{mid} \in N' \text{ or } v_\text{mid}v_\text{out} \in N'}$.
  By construction,~$N$ is disjoint from~$V^\infty$. Let~$P$ be an odd $T$-path in $G$ and let $P'$ be the corresponding $T'$-path in $G'$.
  The path~$P'$ is an odd $T'$-path in~$G'$.
  Since~$N'$ is an odd multiway edge cut in~$G'$, it must contain an edge $e \in E(P')$.
  By choice of $E^\infty$, the edge~$e$ has to be in $\set{v_\text{in}v_\text{mid}, v_\text{mid}v_\text{out} \suchthat v \in V(G) \setminus V^\infty}$.
  Suppose $e \in \set{v_\text{in}v_\text{mid}}$ for some $v \in V(G) \setminus V^\infty$.
  By construction, $v \in N$. Therefore,~$N$ includes a node from every odd $T$-path in~$G$. 
  Thus,~$N$ is a solution to the instance~$\calI$ of the \omwnodec problem.
  Moreover, $\card{N} = \card{N'}$.
\end{proof}

\end{document}